\documentclass[aps,pra,twocolumn,showpacs,amsmath,amssymb,floatfix,superscriptaddress,notitlepage]{revtex4-2}

\usepackage{color}
\usepackage{bbm}
\usepackage{graphicx}
\usepackage{dcolumn}
\usepackage[utf8]{inputenc}
\usepackage{algorithm}
\usepackage{algpseudocode}
\usepackage{afterpage}
\usepackage{graphicx}
\usepackage{epstopdf}
\usepackage{amsthm}
\usepackage{amsmath}
\usepackage{empheq}
\usepackage{bbm}
\usepackage{braket}
\usepackage{amssymb}
\usepackage[usenames,dvipsnames]{xcolor}
\definecolor{tabpurple}{RGB}{146, 104, 172}
\usepackage{cases}
\usepackage{latexsym}

\usepackage{booktabs}
\usepackage{array}

\usepackage[colorlinks=true,citecolor=tabpurple,linkcolor=tabpurple,urlcolor=tabpurple]{hyperref}
\usepackage[title]{appendix}
\usepackage{mathrsfs}

\newtheorem{claim}{Claim}

\newcommand\ddfrac[2]{\frac{\displaystyle #1}{\displaystyle #2}}

\begin{document}

\title{Group word dynamics from local random matrix Hamiltonians and beyond}

\author{Kl\'ee Pollock}
\email{kleep@iastate.edu}
\affiliation{Department of Physics and Astronomy, Iowa State University, Ames, Iowa 50011, USA}
\affiliation{Department of Physics, The Pennsylvania State University, University Park, Pennsylvania 16802, USA}

\author{Jonathan D. Kroth}
\email{jdkroth@iastate.edu}
\affiliation{Department of Physics and Astronomy, Iowa State University, Ames, Iowa 50011, USA}

\author{Jonathon Riddell}
\affiliation{School of Physics and Astronomy, University of Birmingham, Edgbaston, Birmingham, B15 2TT, UK}

\author{Thomas Iadecola}
\affiliation{Department of Physics and Astronomy, Iowa State University, Ames, Iowa 50011, USA}
\affiliation{Ames National Laboratory, Ames, Iowa 50011, USA}
\affiliation{Department of Physics, The Pennsylvania State University, University Park, Pennsylvania 16802, USA}
\affiliation{Institute for Computational and Data Sciences, The Pennsylvania State University, University Park, Pennsylvania 16802, USA}
\affiliation{Materials Research Institute, The Pennsylvania State University, University Park, Pennsylvania 16802, USA}

\begin{abstract}
We study one dimensional quantum spin chains whose nearest neighbor interactions are random matrices that square to one. By employing free probability theory, we establish a mapping from the many-body quantum dynamics of energy density in the original chain to a single-particle hopping dynamics when the local Hilbert space dimension is large. The hopping occurs on the Cayley graph of an infinite Coxeter reflection group. Adjacency matrices on large finite clusters of this Cayley graph can be constructed numerically by leveraging the automatic structure of the group. The density of states and two-point functions of the local energy density are approximately computed and consistent with the physics of a generic local Hamiltonian: Gaussian density of states and thermalization of energy density. We then ask what happens to the physics if we modify the group on which the hopping dynamics occurs, and conjecture that adding braid relations into the group leads to integrability. Our results put into contact ideas in free probability theory, quantum mechanics of hyperbolic lattices, and the physics of both generic and integrable Hamiltonian dynamics.
\end{abstract}
	
\maketitle

\section{Introduction}

Random matrix theory (RMT) has historically provided key insights into the physics of quantum systems in various contexts and limits. Examples include the derivation of universal level spacing statistics of systems with classically chaotic limits \cite{bohigas_characterization_1984}---which also serves as a definition of quantum chaos for systems without classical limits---universal conductance fluctuations in cavities \cite{beenakker_1997}, emergence of two dimensional quantum gravity from randomly interacting quantum particles [the Sachdev-Ye-Kitaev (SYK) model]~\cite{kitaev_simple_2015}, Anderson and many-body localization physics \cite{kravtsov_2015}, and the quantum entanglement and thermal equilibrium properties of random quantum mechanical states \cite{page_1993}. More recently, significant analytical and conceptual progress has been achieved by analyzing dynamics generated by random local quantum gates forming a quantum circuit. Questions such as how entanglement and operators grow dynamically, the interplay between this phenomenon and diffusion of spin density \cite{khemani_operator_2018}, and the general emergence of dissipative semiclassical hydrodynamics from many-body quantum systems can now be treated in broad contexts~\cite{singh_2025}.

In this article we are primarily interested in the the typical dynamics of a conserved energy density in some generic local Hamiltonian system on a 1D lattice. Ultimately, one would like to prove that such a system exhibits energy diffusion. To study this question, we follow the tradition of looking to RMT for some universal insights. However, in moving to RMT, we need to ensure spatial or geometric locality of interactions. Without this, the dynamics of energy density would be somewhat trivial: some excitation of energy would immediately spread everywhere in the system. Another ingredient is total energy conservation so that energy can be transported but not be created or destroyed, so our work must necessarily deviate from the body of literature on random quantum circuits. Once working in RMT, another important observable is the many-body density of states (which governs equilibrium physics), so we also focus on this quantity and study how locality impacts its form as the system grows in size. 

While these requirements are conceptually simple, from an RMT point of view they are highly nontrivial and only recently has there been some progress in the mathematical and physical literature in treating a Hamiltonian RMT with spatially local interactions, even just at the level of the many-body density of states~\cite{bellitti_2019,morampudi_many-body_2019,collins_spectrum_2023,charlesworth_matrix_2021}. We note that a derivation of energy diffusion has been obtained in a chain of coupled SYK dots in~\cite{gu_local_2017}, and the derivation of an emergent Markov process for energy density, which can numerically be seen to be diffusive, in a different local random matrix theory, has been discussed recently in~\cite{chalker_chaotic_2025}.

In the process of treating a particular 1D spatially local random matrix model in which the energy density operators are constrained to square to one, we find that free probability, the theory of infinite discrete groups, and the quantum mechanics of hyperbolic lattices interplay in an interesting way. Free probability is a framework for dealing with non-commuting random variables such as random matrices. In particular, it defines an analogous notion of independence, called free independence, which plays a key role in calculations~\cite{speicher_free_2019}. Such ideas have found applications in quantum chaos theory through the idea of a ``full" eigenstate thermalization hypothesis~\cite{foini_2019,pappalardi_2022}, and more recently in random circuits~\cite{fritzsch_free_2025}. In our work, free probability will again play an essential role in calculations.

After application of free probability, we are brought to the combinatorics of an infinite discrete group, called a right-angled Coxeter group. Coxeter groups are generalized reflection groups which naturally describe the symmetries of tessellations of Euclidean and hyperbolic spaces. There has been recent effort to understand the quantum mechanics of hyperbolic tessellations~\cite{mosseri_2023,lux_spectral_2024,lenggenhager_non-abelian_2023,maciejko_hyperbolic_2021,shankar_hyperbolic_2024} because of applications in error correction~\cite{breuckmann_2016} as well as table-top experiments which effectively simulate quantum mechanics in curved space~\cite{chen_hyperbolic_2023}. A method applicable for the hopping problem on those lattices will also be of use here. We remark that our work is not the first appearance of Coxeter groups in free probability: reference~\cite{fendler_central_2003} derived the limiting distribution corresponding to a central limit theorem for a certain \emph{random} Coxeter group.

\subsection{Summary of results}

Our results and organization are as follows. In Sec.~\ref{sec:hi_rmt_model}A, we introduce a geometrically local random matrix model in 1D whose energy density operators square to one, which we dubbed the Haar-Ising (HI) random matrix model in previous work~\cite{pollock_2025}. We first show, in Sec.~\ref{sec:hi_rmt_model}B, that free probability theory implies a very simple combinatorics for correlation functions of energy density in this model; any infinite temperature energy density correlator of the form $\braket{h_{i_1}\cdots h_{i_n}}$ is either unity or vanishing in the limit of large Hilbert space dimension $q$. 

Due to the geometry of the model, i.e. nearest neighbor interactions on a chain with periodic boundary conditions (PBC), the resulting combinatorics of those correlators is captured precisely by a right-angled Coxeter reflection group, which we discuss in Sec.~\ref{sec:tools}A. The Cayley graphs for these groups are infinite undirected regular graphs and, for some small cases, namely length $L=3,4,5$, PBC RMT chains correspond to tessellations of the hyperbolic disk by polygons. This is discussed in Sec.~\ref{sec:tools}B, but we leave open the question of their geometry for larger $L$, where higher-dimension spaces will be necessary. In order to make calculations on these graphs possible, one needs a way to enumerate and order unique group elements in an efficient way. This is possible to do when the group admits an \emph{automatic structure}, which we introduce and discuss in Sec.~\ref{sec:tools}C.

Next, we take the combinatorial correspondence between the local HI chain and certain right-angled Coxeter groups further. We show that an arbitrary space-time correlation function of energy density in the original model can be written as a \emph{single-particle hopping process} on the corresponding Cayley graph. This represents a kind of random-matrix/single-particle duality that is our main conceptual result. Since the particle hops on the Cayley graph of a group, we refer to this as \emph{group word dynamics}. A set up for the group word hopping problem and the above duality appears in Secs.~\ref{sec:hopping}A,B. We remark that a somewhat different kind of group word dynamics has also been recently discussed in~\cite{balasubramanian_2024}. In that work, group words are configurations in a many-body Hilbert space on the lattice and they consider a dynamics that is kinetically constrained by relations in the group. In our work, the dynamics under the group is single-particle.

The group word hopping occurs on a hyperbolic lattice that lacks abelian translation symmetry, which makes usual momentum space methods inapplicable. This is where the automatic structure of the group can be exploited. In Sec.~\ref{sec:hopping}C we outline an efficient recursive algorithm that uses the automatic structure to construct huge adjacency matrices for the Cayley graphs. This is necessary for convergence since, unlike Euclidean lattices, there are a large fraction of sites at the boundary in open boundary conditions (OBC). This algorithm represents our main technical result.

Using these adjacency matrices, in Sec.~\ref{sec:calculations} we first study the density of states of various HI random-matrix chains up to length $L=14$ and observe how they slowly evolve from RMT-like forms with sharp band edges towards a Gaussian with tails, consistent with spatial locality of interactions. We then study two-point functions $\braket{h_i(t)h_j(0)}$ of energy density for primarily the $L=5$ chain, where the numerical results are best converged. We see how an excitation of energy density on some bond spreads to other bonds in a manner consistent with the locality of the interactions. We also observe thermalization dynamics, with correlators approaching $L^{-1}$ after an $O(1)$ amount of time as $q\rightarrow \infty$, suggesting that the large-$q$ limit is physically meaningful in this model.

Finally, in Sec.~\ref{sec:integrable} we ask if the group word hopping picture can be interesting in other contexts. To this end, we ask what happens if we formally modify the group from the right-angled Coxeter group originally introduced, to something ``smaller." In particular, we consider adding braid relations into the group. The braid relations are reminiscent of recently discussed “parameter-less” or “set-theoretic” Yang-Baxter relations in the context of cellular automata or circuits \cite{gombor_2022}. We observe that with the braids, the dynamics of energy density becomes less generic and more special, for example displaying pronounced fluctuations. Adding one more relation into the group, we get the symmetric group $S_L$ for a length $L$ chain, leading to interesting exact revivals of energy density for $L=4,5$ chains. We also find that braids allow us to construct at least two more nontrivial local conserved quantities beyond energy at the level of the group algebra using the so-called boost operator. We see no obstruction to generating extensively many such charges and this leads us to conjecture the system to be integrable.

We conclude in Sec.~\ref{sec:conclude} by introducing more examples of the group word hopping. One can turn the RMT/single-particle duality around: given a hyperbolic tesselation, can we construct a ``matrix model" for this hopping problem? We give the example for a $\{4,6\}$ tiling, which is not included in the PBC nearest neighbor geometries so far discussed.

\begin{figure}
    \centering
    \includegraphics[width=0.8\linewidth]{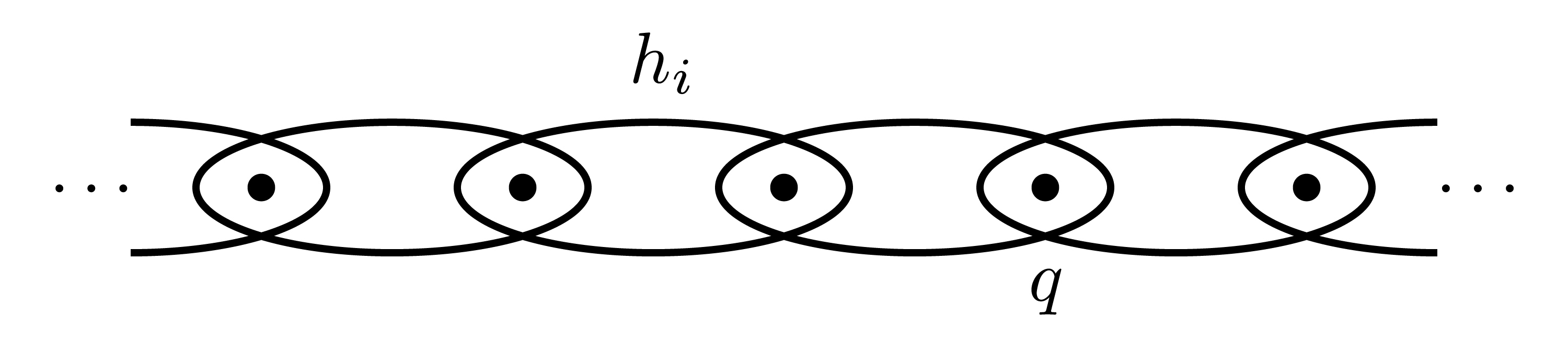}
    \caption{Geometry of the local Haar-Ising random matrix spin chain.}
    \label{fig:local_chain}
\end{figure}

\section{Local Haar-Ising random matrix model}\label{sec:hi_rmt_model}

\subsection{Model}

We consider the following ensemble of Hamiltonians acting on a chain of $L$ sites, each with local Hilbert-space dimension $q$:
\begin{equation}
    H = \sum_{i=1}^L h_i,
\end{equation}
where $h_i$ are independent random matrices coupling sites $i$ with $i+1$ and so act non-trivially only on the Hilbert space $\mathbb{C}^q \otimes \mathbb{C}^q$ corresponding to those sites, as in Fig.~\ref{fig:local_chain}. We focus exclusively on periodic boundary conditions (PBC) throughout, whereby $L+1 \equiv 1$. We adopt the conditions that $h_i^2=I$ and $\text{tr} h_i = 0$. Beyond these constraints, we take the eigenvectors to be maximally ergodic within their local Hilbert space by setting
\begin{equation}
\label{eq:def_local_hi}
    h_i = U_i \Lambda U_i^\dagger \otimes I_{\overline{i, i+1}},
\end{equation}
where $U_j$ are independent and Haar-randomly drawn from the unitary group $\mathcal{U}(q^2)$, while $\Lambda^2=I$, $\text{Tr}(\Lambda)=0$, and $I_{\overline{i, i+1}}$ denotes an identity matrix acting on all sites not $i,i+1$. We call the random matrices generated in this way ``Haar-Ising" since the local terms have eigenvalues constrained to be $\pm 1$.  This model is interesting because it admits a particularly simple combinatorics, while maintaining all of the crucial aspects of a true spatially local random matrix Hamiltonian.

\begin{figure*}
    \centering
    \includegraphics[width=0.9\linewidth]{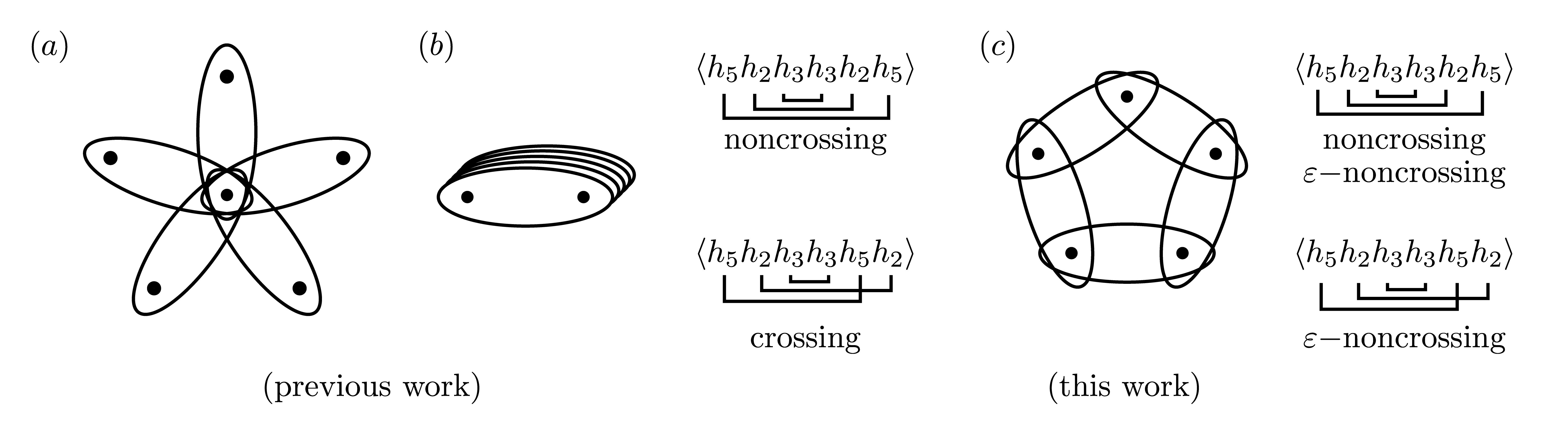}
    \caption{Various random matrix models involving $5$ HI bonds. Examples leading to standard freeness are illustrated in (a) and (b). In (a), the Hamiltonian is more ``local" than (b) in the sense that it is vanishingly sparse as $q\rightarrow\infty$ but neither are geometrically local as far as energy density operators are concerned. Panel (c), however, is an $L=5$ geometrically local chain corresponding to $\varepsilon$-freeness for $\varepsilon$ defined as in Eq.~\eqref{eq:varepsilon}; this type of geometry is the focus of the paper.}
    \label{fig:free_vs_eps_free} 
\end{figure*}

\subsection{Free probability in the HI context}

Here, we give a brief exposition of some ideas from free probability and $\varepsilon$-free probability, but specialized to our particular setting where $h_i^2=I$. In general, free probability treats random variables which do not commute. This topic can be discussed in an algebraic way~\cite{nica_lectures_2006} or in the context of random matrices~\cite{mingo_free_2017}, which serve as examples of free variables when the matrix dimension goes to infinity; our primary interest in this section is the latter. Our goal is to make statements about infinite temperature equal-time correlation functions of the form
\begin{equation}\label{eq:corr}
    \braket{h_{i_1} \cdots h_{i_n}},
\end{equation}
where $i_1 \cdots i_n$ is an arbitrary word of indices with each $i_j \in [L]$, where we define the set $[L]=\{1,2,\dots ,L\}$. The average $\braket{\cdot}$ is defined as
\begin{equation}\label{eq:trace}
    \braket{\cdot} = q^{-L}\ \mathbb{E}\ \text{Tr}(\cdot)
\end{equation}
and the average $\mathbb{E}$ is over all present terms. In our case, this is taken with respect to the uniform (Haar) measure on $\mathcal{U}(q^2)^L$, through the eigenvectors $U_i$ of each term.

\subsubsection{Standard free independence}

Let us first consider what standard free probability and random matrix theory has to say about equal-time infinite temperature correlation functions (also known as mixed moments) of the form Eq.~\eqref{eq:corr} when all terms fail to commute. In Fig.~\ref{fig:free_vs_eps_free}, models (a) and (b), constructed out of 2-body HI interactions, are examples where all terms mutually do not commute. It turns out the collection of terms in these models are \emph{asymptotically freely independent} as $q\rightarrow\infty$. Here, asymptotic freeness is a particularly simple statement which demands that
\begin{equation}
    \braket{h_{i_1} \cdots h_{i_n}} \rightarrow 0\quad  \text{as} \quad q\rightarrow \infty
\end{equation}
whenever the index word $i_1\cdots i_n$ has an odd number of any particular $i$ present, or when the word has only an even number of each $i$ present, but contains a \emph{crossing}. By a crossing, we mean the following: given a group word, draw lines that group together all identical symbols below the word, as in Fig.~\ref{fig:free_vs_eps_free} (i.e. draw the partition of $[n]$ induced by the index word). If all groups can be drawn in a planar way (i.e. the lines can be drawn so that no two lines cross), the word is called non-crossing. Otherwise it is crossing.

\subsubsection{Notions of locality}

The fact that the terms are freely independent in (a) follows from our previous work~\cite{pollock_2025} where we showed that two bonds overlapping on a single site is sufficient to induce freeness as $q\rightarrow\infty$ due to the Haar-random eigenvectors. That the same holds for model (b) is almost a standard result in free probability~\cite{speicher_free_2019}. Freeness in both scenarios (a,b) and the $\varepsilon$-free scenario we will next discuss are all subsumed by the results obtained in Ref.~\cite{collins_spectrum_2023}. From the point of view of an observable involving only energy density operators, these two models are indistinguishable for $q\rightarrow \infty$, for example the density of states or $n$-point functions of energy density, and we derived such expressions for $z$ terms arranged as in (a) or (b) in our previous paper. However, neither are geometrically or spatially local from this point of view: an excitation of energy density at some bond will immediately spread to all other bonds. Furthermore, if we consider $z$ such terms arranged in either way, for large $z$, we will obtain a semicircular density of states of width $\sqrt{z}$ as opposed to a Gaussian of width $\sim\sqrt{z}$, the result expected for a many-body system with short-range interactions. There is a sense in which (a) is ``more local" than (b), however; (a) is vanishingly sparse as $q\rightarrow\infty$, unlike (b), which is dense, so the physics may be different at finite $q$.

\subsubsection{\texorpdfstring{$\varepsilon$-free independence}{epsilon-free independence}}

We now contrast standard freeness with the case of interest in our paper: a nearest-neighbor spin chain where $h_i$ and $h_j$ commute when $|i-j|>1$, and since they are drawn independently they behave as classically independent random variables. Neighboring bond operators $h_i$ and $h_{i+1}$, however, are still freely independent. In particular, we have
\begin{equation}
    \braket{(h_i h_{i+1})^m} \rightarrow 0
\end{equation}
for any $m\geq 1$. This scenario, which combines classical and free independence, falls under the umbrella of $\varepsilon$-freeness in the mathematical literature. For a collection of $L$ variables, one expresses this via a symmetric $L \times L$ matrix $\varepsilon$ that contains a unit entry $\varepsilon_{ij} =1$ if $h_i$ and $h_j$ are classically independent for $i\neq j$ and a zero entry $\varepsilon_{ij} =0$ if $h_i$ and $h_j$ are freely independent for $i\neq j$. The diagonal entries are meaningless. The $L=5$ chain with periodic boundary conditions as in Fig.~\ref{fig:free_vs_eps_free}(c) has
\begin{equation}\label{eq:varepsilon}
    \varepsilon = 
    \begin{pmatrix}
    0 & 0 & 1 & 1 & 0 \\
    0 & 0 & 0 & 1 & 1 \\
    1 & 0 & 0 & 0 & 1 \\
    1 & 1 & 0 & 0 & 0 \\
    0 & 1 & 1 & 0 & 0
    \end{pmatrix}.
\end{equation}
Note that setting all non diagonal entries to one corresponds to standard freeness. 

One can again establish for this model a similar statement about correlation functions. We will call a word $\varepsilon$-crossing if it contains a crossing between letters that do not commute. Now,
\begin{equation}
    \braket{h_{i_1} \cdots h_{i_n}} \rightarrow 0\quad  \text{as} \quad q\rightarrow \infty
\end{equation}
whenever the index word $i_1\cdots i_n$ has an odd number of any particular $i$ present, or when the word has only an even number of each $i$ present, but is $\varepsilon$-crossing. See Fig.~\ref{fig:free_vs_eps_free}(c) for an example of a word which is $\varepsilon$-noncrossing but crossing.

Consequently, under the infinite temperature random matrix average, we have for an arbitrary word $i_1 \cdots i_n$ that
\begin{equation}
    \braket{h_{i_1} \cdots h_{i_n}} = 
    \begin{cases}
        1\quad \mathrm{if}\quad h_{i_1} \cdots h_{i_n} = I\\
        0 \quad \mathrm{otherwise}
    \end{cases}
\end{equation}
in the limit of $q\rightarrow \infty$. That is to say, if the operators present in the correlation function can be shuffled around (subject to the constraint that the nearest neighbors fail to commute) to obtain the identity matrix, the correlator will survive a large-$q$ average. Otherwise, it does not contribute anything in that limit.

\section{Tools from geometric and computational group theory}\label{sec:tools}

\begin{figure*}
    \centering
    \includegraphics[width=0.87\linewidth]{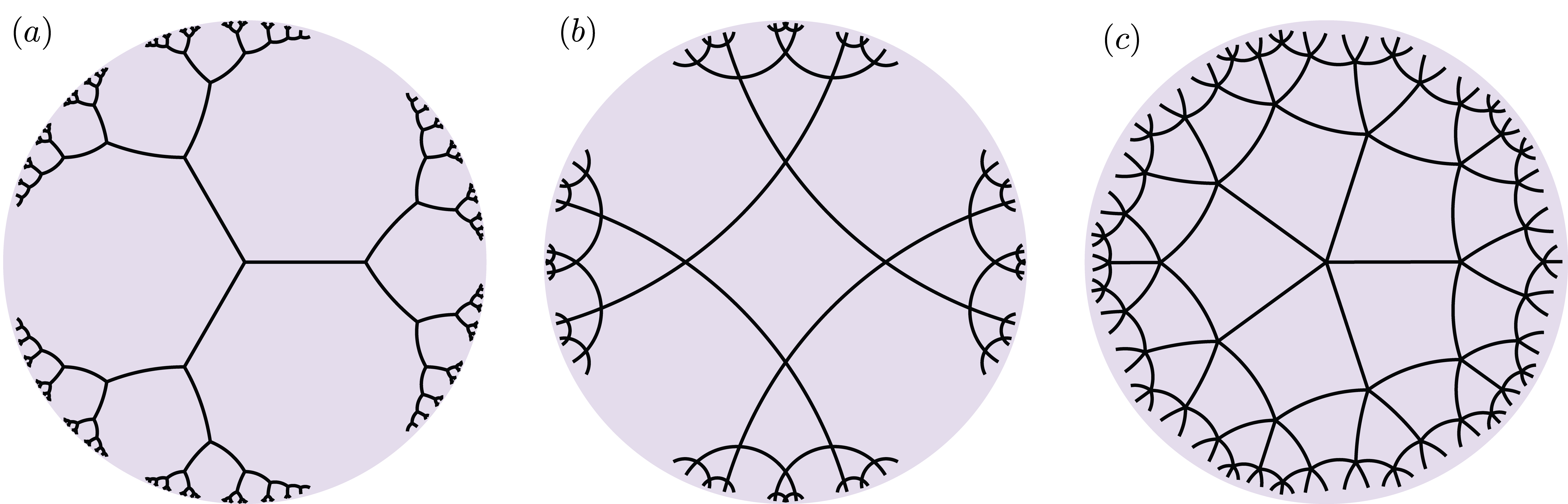}
    \caption{Finite portions of the Cayley graphs of $G_L$, with respect to the generators $g_i$ for $L=3,4,5$ in panels $(a),(b),(c)$, respectively. The graphs are embedded in the hyperbolic disk to emphasize their hyperbolic nature. The graphs can thus be viewed as tilings of the hyperbolic disk by polygons. In (a), $3$ \emph{apeirogons} meet at each vertex, in (b) two squares and two apeirogons meet, and in (c) five squares meet.}
    \label{fig:graphs}
\end{figure*}

\subsection{The Coxeter group of interest}

Let us take a slight detour away from the random matrix problem in order to introduce some ideas from the theory of infinite discrete groups. Consider the group $G_L$ generated by $\{g_j\}_{j=1}^L$ obeying certain relations inherited from the original asymptotic random matrix problem or equivalently the abstract $\varepsilon$-free algebraic problem. The group for the 1D nearest neighbor HI chain is presented by
\begin{equation}\label{eq:group}
    G_L = \braket{\ g_j\  |\ g_j^2,\  (g_i g_j)^2\ \text{if}\ |i-j|>1\ }, 
\end{equation}
the notation meaning that we have a group with generators $\{g_j\}_{j=1}^L$, which can be considered initially as the free group (no relations) on those generators. We then impose certain relations. Here, we are imposing the relations
\begin{equation}
    g_i^2 = e,\quad g_i g_j = g_j g_i \ \text{when}\ |i-j|>1.
\end{equation}

Let us make a few observations about this group. First, it is infinite: there are infinitely many words of the form $(g_ig_{i+1})^k$, for example. This is what is meant by freeness of $g_i$ and $g_{i+1}$: there is no further relation between them. The subgroup of $G_L$ generated by $g_i$ and $g_{i+1}$ is the group
\begin{equation}
    \braket{ g_i, g_{i+1} } = \mathbb{Z}_2 *\mathbb{Z}_2.
\end{equation}
Here, $G * H$ is the free product of groups $G$ and $H$: the group of words formed by generators of $G$ and $H$ without further relations beyond those inherited from $G$ and $H$ individually.

Second, the family of groups $G_L$ are examples of Coxeter groups, which are discrete groups with at most pairwise relations among a finite set of otherwise free generators which square to the identity. Coxeter groups were introduced as generalizations of the groups of symmetries of regular polytopes. $G_L$ are also right-angled Coxeter groups, meaning the only relation between generators is the specification of whether or not they commute. In Sec.~\ref{sec:integrable} we will consider other Coxeter groups which are not right-angled. Combinatorial~\cite{kolpakov2020} and geometric aspects~\cite{dani_large-scale_2018} of right-angled Coxeter groups are still active topics in the mathematical literature.

\subsection{The Cayley graph}

Given a group $G$ and a set of generators $S$ of $G$, the Cayley graph of $G$ with respect to $S$ is the following graph~\footnote{Cayley graphs are often defined as directed graphs. In the Coxeter group context, where generators are self-inverse, this is unnecessary.}: for each group element $g \in G$, there is a node. Two nodes corresponding to $g,h \in G$ are then connected by an edge if there is a $g_i \in S$ such that $g = g_i h$ or $g = h g_i$. The simplest group in the family Eq.~\eqref{eq:group} has $L=2$, for which the Cayley graph is a one dimensional infinite chain, which we implicitly discussed in Ref.~\cite{pollock_2025}. The structure is more interesting for larger $L$. Fig.~\ref{fig:graphs} shows the graphs for $L=3,4,5$, which are planar. 

As an aside, we remark that when the Cayley graph of a Coxeter group is embedded in an appropriate metric space, the Coxeter group is a reflection group in the following sense: given a node $h$ and a generator $g_i$, the action of $g_i$ on $h$ is to reflect $h$ across the hyperplane perpendicular to the edge connecting $h$ and $g_i h$ into $hg_i$. In Figure \ref{fig:graphs}, the hyperplanes are curves meeting the edges of the graphs at right angles (with respect to a hyperbolic metric).


The $L=3$ Cayley graph is a tree; in that case, $\varepsilon$-freeness corresponds to standard freeness. Geometrically, the Cayley graph forms a $\{p,q\}=\{\infty,3\}$ tiling of the hyperbolic disc. Here, the Schl\"afli symbol $\{p,q\}$ refers to a tiling of the plane by $p$-gons, where $q$ such $p$-gons meet at every vertex. If $(p-2)(q-2) = 4$, this is the Euclidean plane. If $(p-2)(q-2) > 4$, it is the hyperbolic plane. The $L=4$ graph is also known as a Husimi cactus graph~\cite{mosseri_bethe_1982}, a tiling of the plane by squares and apeirogons. The $L=5$ graph is a regular tiling of the hyperbolic disk with $\{p,q\} = \{4,5\}$. For $L\geq 5$, the group does not decompose into a free product $A*B$~\cite{office_hours_2017}. This rules out the use of the well-developed free-convolution method of free probability, which produces the Green's function of a sum of two free variables~\cite{speicher_free_2019}. We note that, for generators which square to one, the statement that $g_i$ and $g_j$ commute corresponds to a loop of length $4$. In general, each node of the Cayley graph will participate in $L(L-3)/2$ loops for $L\geq 3$. Therefore, for $L\geq 6$ the Cayley graphs will not embed into a plane, and we leave a study of their precise geometry for future work.

Going forward, we define the \emph{radius} $d$ of an element $g\in G_L$ as the graph distance going from $g$ to $e$ on the Cayley graph; that is, the length of a geodesic path from $g$ to $e$. As will become clearer in the next section, this is also the length of a shortest word written in the generators $g_i$ which equals $g$ as a group element.

\subsection{The automatic structure}

\begin{figure}
    \centering
    \includegraphics[width=0.85\linewidth]{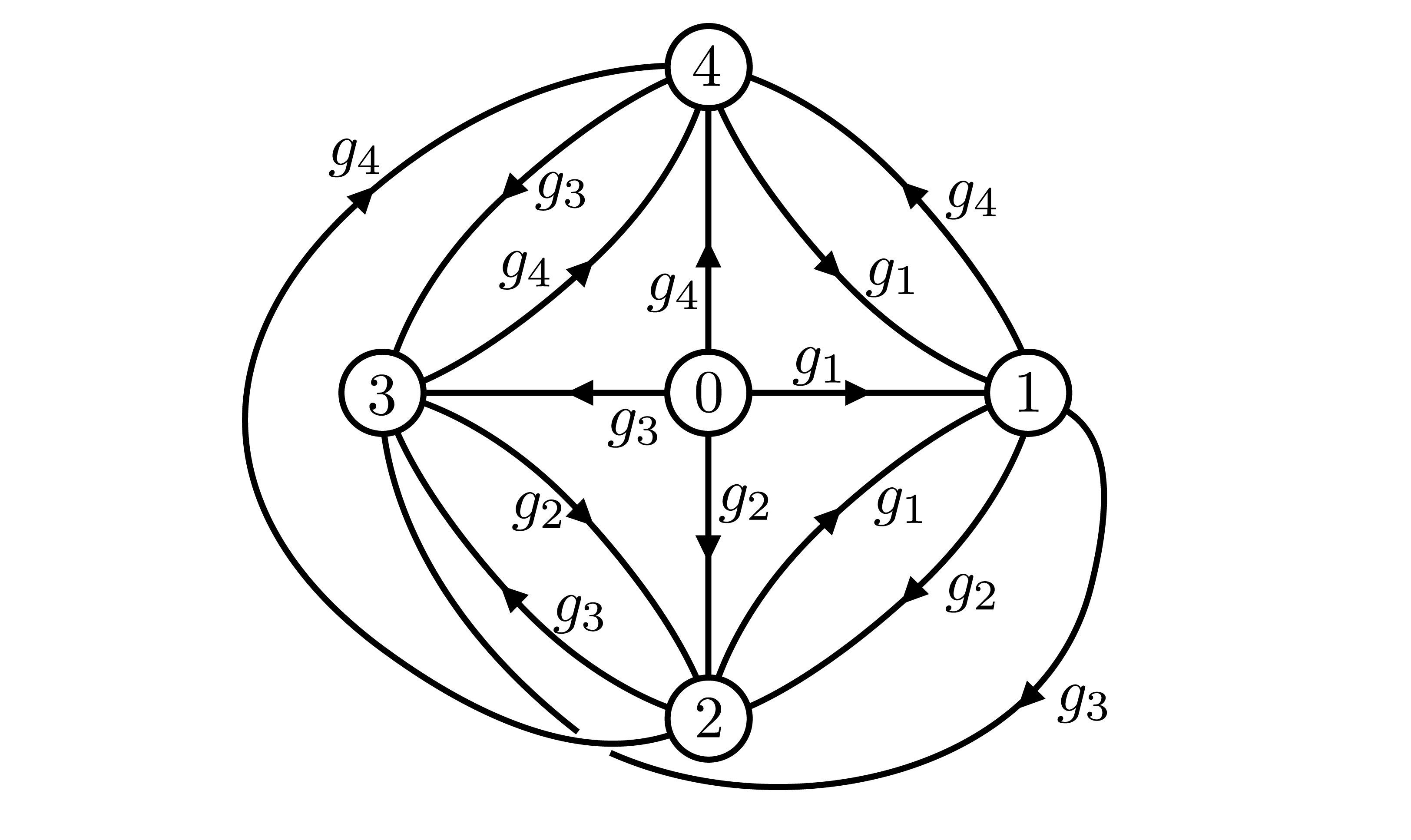}
    \caption{The finite state automaton $W$ which recursively enumerates all words in the language $\mathcal{L}$ of short-lex geodesic words for the group $G_4$ corresponding to an $L=4$ local HI RMT chain in PBC. The nodes of the automaton are states and transitions correspond to concatenation of the letter $g_i$ to a previously allowed word $g\in \mathcal{L}$. The recursive step consists of generating all words of radius $d+1$ by concatenating all words of radius $d$ with a single generator $g_i$. Given the state of a radius $d$ word, $W$ dictates which letters $g_i$ can be concatenated to form radius $d+1$ words and the corresponding states. See the text for further details.}
    \label{fig:automaton}
\end{figure}

Elements of the discrete group $G_L$ are represented by abstract group words in the generators. To perform calculations, however, it will be crucial to identify a unique word of minimal length in the generators which represents each group element. In the computational group theory literature, the set of such words in the letters $g_i$ is called the language $\mathcal{L}$ of \emph{short-lex geodesic group words}. Given a group element $g$, one identifies the corresponding element of $\mathcal{L}$ by (1) assigning the shortest possible word in the letters $g_i$ which equals $g$ as a group element and (2) further choosing the first word appearing in the lexicographic order $g_1 > g_2 > ... > g_L$. For example, in $G_{L\geq 4}$, $g_1g_3 = g_3g_1$, but we choose to use $g_1g_3$ to identify this group element since $g_1 > g_3$ lexicographically. We make this comparison from left to right until a difference is found.

The definition of the short-lex geodesic words above is, however, not constructive: it is computationally prohibitive to produce all possible strings in the letters $g_i$ and try to reduce them to their minimal representative. A more efficient way to enumerate the words is to look to the group's \emph{automatic structure}. Among other things, this means that there exists a finite state automaton $W$, called the \emph{word acceptor automaton}, which can be used to recursively enumerate the language $\mathcal{L}$ of short-lex geodesic words. More specifically, given a length $d$ word in $\mathcal{L}$, the automaton can produce a length $d+1$ word with a cost that is independent of $d$. Not all finitely presented discrete groups are automatic as such, however, all finitely presented Coxeter groups are~\cite{holt_2023} and this is what makes calculations possible in our case.

The word acceptor automaton is best explained with an example. $G_4$ is a non-trivial example but has a reasonably simple $W$. The group is
\begin{equation}
    G_4=\braket{\ g_i\ |\ g_i^2,(g_1g_3)^2,(g_2g_4)^2\ }.
\end{equation}
The word acceptor automaton has, in this case, five states and 14 transitions, see Fig.~\ref{fig:automaton}. The word acceptor automaton dictates which letters $g_i$ can be recursively appended to enumerate all possible short-lex geodesic group words in $\mathcal{L}$. The recursion works as follows, starting on the trivial node labeled $0$ in the figure and corresponding to the trivial group word $e$. For each allowed transition, we concatenate the letter $g_i$ to $e$. We have produced so far the set $\{e,g_1,g_2,g_3,g_4\}$, i.e. all words up to and including radius $1$. In general, starting from a group word $g\in \mathcal{L}$ with radius $d$ and a state $s$ of $W$, which we also track, we build all new words formed by allowed transitions from $s$ by concatenation of the corresponding label. Following the recursion, one finds that the first few words in the resulting short-lex geodesic language are
\begin{multline}
    \mathcal{L} = \{e, 1,2,3,4,12,13,14,21,23,24,32,34,41,43,\\
    121,123,124,132,134,\dots\},
\end{multline}
where we have introduced the shorthand $ijk\dots \equiv g_ig_jg_k\dots$.
Each group word $g$ in this language, when read from left to right, corresponds to a geodesic path from $e$ to $g$ on the Cayley graph. The fact that the group words can be enumerated as such will ultimately allow us to construct huge adjacency matrices for Cayley graphs with $\sim 10^7$ nodes. We state the algorithm that does so in the following section.

Generation of the word acceptor automata was done using the KBMAG package in the computer algebra software GAP~\cite{holt_2023}. For the interested reader, we have linked our GAP code in a GitLab project \cite{adjacency_racg_2025}.

\section{Group word hopping}\label{sec:hopping}

\subsection{Setting up the hopping problem}

We are interested in computing single-trace infinite temperature random matrix averages in the local HI chain. We will ultimately see that any such quantity (static or dynamical) which involves only the energy density operators will admit a description in terms of a single particle hopping on the Cayley graph. Let us first introduce the players of the hopping problem on the Cayley graph of $G_L$ and then see how the correspondence with the original RMT model arises.

As discussed in Sec.~\ref{sec:tools}, the nodes on the Cayley graph correspond to group elements. Let us denote by $\ket{g}$ the state with a single particle sitting at node $g$. A general single particle wavefunction then has the form
\begin{equation}
    \ket{\psi} = \sum_{g\in G_L} \psi_g \ket{g}.
\end{equation}

To describe the hopping problem, we will also need to define the action of an arbitrary element $h\in G_L$ on this vector space. The most convenient way to is consider a (Hermitian and unitary) matrix representation $V_h$ of $h$ which acts as
\begin{equation}
    V_h = \sum_{g\in G_L} \ket{gh}\bra{g}.
\end{equation}
When $h$ is one of the group generators $g_i$, we write $V_i$ for its matrix representation. The hopping problem is then fully determined by the adjacency matrix $\Delta$ of the graph, given by
\begin{equation}
    \Delta = \sum_{i=1}^L V_i = \sum_{g\in G_L} \sum_{i=1}^L  \ket{gg_i} \bra{g}.
\end{equation}
This is simply the statement that from any node $g$, the particle can hop to any nearest neighbor with uniform weight.

\subsection{Random-matrix/single-particle duality}

Having set up all of the ingredients, we now state the correspondence between single-trace random matrix averages and the hopping problem. To go from one to the other, make the following replacements:
\begin{gather}\label{eq:correspondences}
    h_i\quad \longleftrightarrow\quad V_i\\
    H\quad \longleftrightarrow\quad \Delta\\ 
    \braket{\cdot}\quad \longleftrightarrow\quad \braket{e|\cdot|e}
\end{gather}
Let us demonstrate why these yield equivalent results for the density of states. This is essentially because powers of $H$ are a sum of all possible words of length $n$:
\begin{equation}
    \braket{H^n} = \sum_{i_1 \cdots i_n} \braket{h_{i_1} \cdots h_{i_n}}
\end{equation}
and, in the large $q$ limit, $\braket{h_{i_1} \cdots h_{i_n}}$ is either zero or one according to whether or not the word is reducible to the identity. This is precisely the statement that $\braket{h_{i_1} \cdots h_{i_n}} = \braket{e|V_{i_1} \cdots V_{i_n}|e}$ and so $\braket{H^n} = \braket{e|\Delta^n|e}$. Since the single-particle hopping problem and the random-matrix problem have the same sequence of moments, they have the same density of states~\footnote{This is because both have bounded support, i.e. no eigenvalues lie outside the range $|\epsilon| \leq L$}. While this type of correspondence is standard~\cite{nica_lectures_2006}, a conceptually non-trivial extension of this idea is to treat arbitrary space-time dependent correlation functions. We can also establish for arbitrary space-time points $(i_1,t_1),\cdots, (i_k, t_k)$ the duality
\begin{equation}
    \braket{h_{i_1}(t_1) \cdots h_{i_k}(t_k)} = \braket{e|V_{i_1}(t_1) \cdots V_{i_k}(t_k)|e}
\end{equation}
where, on the right-hand side, time evolution is implemented by the hopping Hamiltonian: $V(t)=e^{-i\Delta t} V e^{i\Delta t}$. This also follows by using that $H^n$ is a sum of all possible words, but this fact is somewhat deeper since it states that the average large-$q$ quantum many-body dynamics of energy density can be described entirely in terms of a single particle hopping dynamics. While this is a conceptually interesting result, one might wonder if we have only recast the difficult random matrix problem as another difficult problem. It turns out, however, that the automatic structure of the group can be used to obtain some highly nontrivial numerical solutions directly in the $q\rightarrow \infty$ limit.


\begin{figure}
    \centering
    \includegraphics[width=0.9\linewidth]{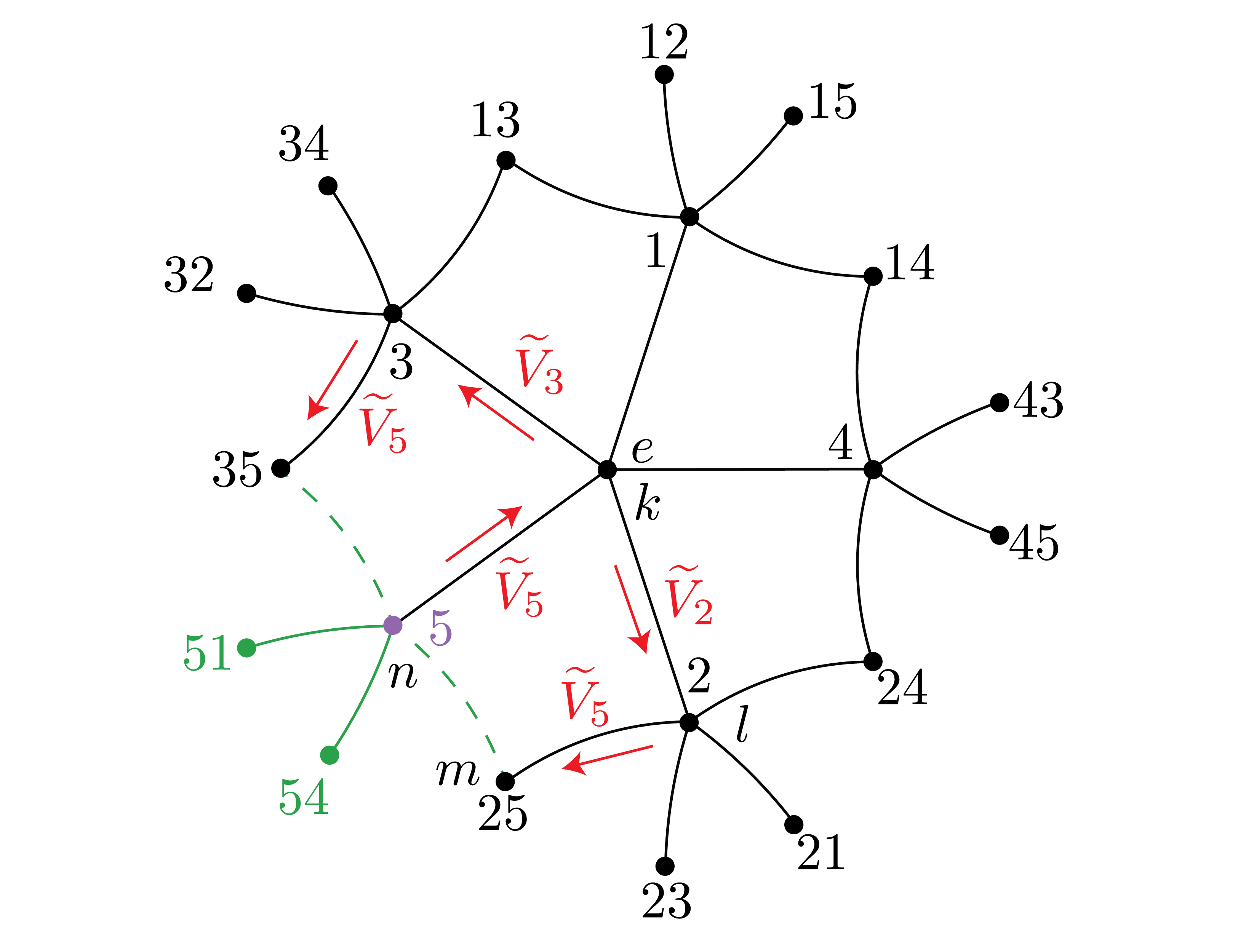}
    \caption{The recursive algorithm producing short-lex geodesic group words and adjacency matrices $V_i$ for the group $G_5$ corresponding to an $L=5$ HI RMT chain. Here we are performing the recursive step on the purple node, the group word $g_5\equiv 5$, and all black nodes and connections have been created. We need to put unit entries into $V_1,V_2,V_3,V_4$ in the correct locations. The green nodes and solid green connections are easily made, as they correspond to new group words generated by the word acceptor $W$. The dashed green connections are more subtle; we need to connect node $n$ to $m$, for example. Here, $n,m,k,l$ are indices for the location of the corresponding group word in the ordered set $\mathcal{L}$, not shorthand for generators. Starting on node $g_5$, $W$ will not allow a transition to $g_5g_2$ since $g_2g_5$ is already an element of $\mathcal{L}$ which represents that group element. Since the transition is rejected by $W$, we know it must be part of a cycle of length $4$ and we can use the previously constructed adjacency matrices to identify the index $m$ that needs to be connected to $n$.}
    \label{fig:algo}
\end{figure}

\subsection{Numerical construction of adjacency matrices}

Our main technical result is an algorithm to numerically construct the matrix representations $V_i$ on large but finite clusters of the Cayley graph of $G_L$ by taking advantage of the group's automatic structure. Fig.~\ref{fig:algo} demonstrates the working principle of the algorithm for $L=5$. The algorithm is recursive and the recursion step is as follows. We are sitting at a node at a radius $d$ from the node $e$. It corresponds to an element $g\in G_L$, which in turn corresponds to the $n$th element of the ordered set $\mathcal{L}$. In the example, this is the purple node corresponding to the group word $g_5$, written $5$ for short. In that example, this is element $6$ in $\mathcal{L}$.

The first task is to enumerate new elements of $\mathcal{L}$, i.e. children of $g$ which are at radius $d+1$ from the origin, and assign to each child an index in the short-lex ordering. Those are the words $gg_i$ which are accepted by the word acceptor. In the above example, these are the words $51$ and $54$. The number of accepted generators depends locally on $g$. Another crucial step is to record the state of the word acceptor automaton at the creation of each new word. In this way, when we recursively end up at the new word, we will only need to perform one additional transition.

The second task is to record all connections $(n,m)$ from the current node $n$ to adjacent nodes $m$ of radius $d+1$ into the adjacency matrices. These are the green lines in the example. If the connection is made via generator $g_i$, that entry goes into $V_i$. Some connections are made by the automaton (solid green), of the form $(n,n+1),(n,n+2),\dots$. However, some connections (dashed green) are not made by the automaton, since $m$ may be a child of a previous word, e.g. $25$ is a child of $2$, not $5$. These correspond to loops in the Cayley graph, so we know that there is a length $3$ path taking us to the desired node in $\mathcal{L}$.

To make the connections $(n,m)$ in that instance, we use the adjacency matrices so far constructed, say $\widetilde{V}_i$. They will already contain all necessary connections due to the short-lex ordering. Suppose the word acceptor rejects $2$ as in the example. Since the starting node was $g=5$, we know that we will need to move along the walk $525$ in order to extract the index $m$ of the node that should be connected to $n$. This is accomplished by using the corresponding matrix representations of $g_i$ so far constructed:
\begin{equation}
    m \xleftarrow[]{\widetilde{V}_5} l \xleftarrow[]{\widetilde{V}_2} k \xleftarrow[]{\widetilde{V}_5} n
\end{equation}
and we may then enter $(n,m)$ into $V_2$ if it doesn't yet exist \footnote{In practice, we only construct the upper triangular part of all $V_i$, but then since the graph is undirected, we add the transpose at the end.}. This summarizes our method for construction of the adjacency matrices $V_i$. Our Python code to produce these matrices is contained in the GitLab project \cite{adjacency_racg_2025}.

\begin{figure*}
    \centering
    \includegraphics[width=0.95\linewidth]{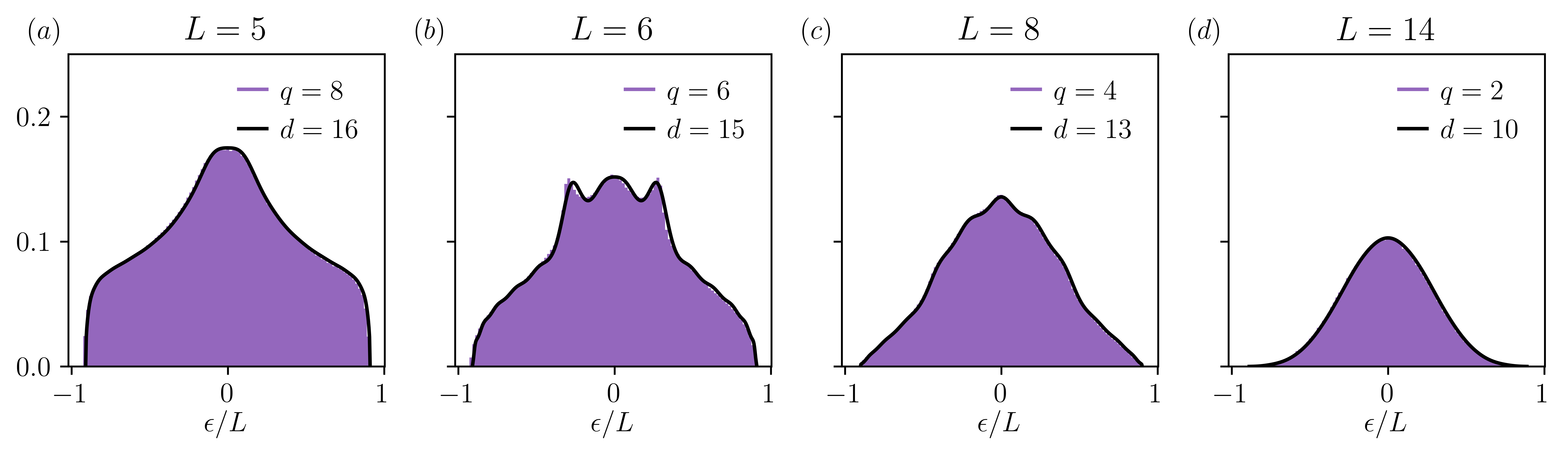}
    \caption{Many-body density of states $\mathcal{D}(\epsilon)$ for local HI RMT spin chains of length $L$ with PBC. Black curves are obtained from the single-particle hopping problem on the corresponding Cayley graph up to radius $d$. Purple histograms are obtained from the numerically exact eigenvalues of a single instance of the random matrix spin chain with local Hilbert space dimension $q$. We show both the largest $q$ and $d$ we were able to access by the corresponding algorithms.}
    \label{fig:density_of_states}
\end{figure*}

\section{Calculations}\label{sec:calculations}

\subsection{Density of states}
\label{sec:dos}

To compute the (local) density of states on the Cayley graph, we first obtain the local Green's function
\begin{equation}
    G(z) = \braket{e|(z-\Delta)^{-1}|e}.
\end{equation}
To approximate it, we employ the Lanczos tridiagonalization method. This method iteratively constructs a basis $|n)$ starting with $|1) = \ket{e}$ in which the Hamiltonian $\Delta$ is tridiagonal, having the form
\begin{equation}
    \Delta = \sum_{n=1}^\infty \sqrt{b_n} \bigg(|n)(n+1| + |n+1)(n| \bigg).
\end{equation}
The \emph{Lanczos coefficients} $b_n$ are the output of the algorithm, and are used to provide a better and better approximation to the density of states. As the iterative algorithm proceeds, the coefficients settle down to some value $b_\infty$. That they can approximate the local Green's function follows because for the above tridiagonal Hamiltonian, $G(z)$ takes on the form of an infinite continued fraction~\cite{haydock_1972},
\begin{equation}
    G(z) = \ddfrac{1}{z-\ddfrac{b_1}{z-\ddfrac{b_2}{z- \cdots}}} .
\end{equation}
For a tree with coordination number $m$, the algorithm becomes exact since the coefficients settle down after one iteration: $b_1 = m$ and $b_{n\geq 2} = m-1$ \cite{mahan_energy_2001}. The remaining continued fraction expansion takes on a closed form and the Green's function becomes
\begin{equation}
    G_{\mathrm{tree}}(z) = \ddfrac{1}{z-\ddfrac{2m}{z+\sqrt{z^2-4(m-1)}}} ,
\end{equation}
which can be seen to be equivalent to the standard expression for a tree~\footnote{The standard expression is
\begin{equation}
    \frac{2(m-1)}{(m-2)z-m \sqrt{z^2-4(m-1)}}.
\end{equation}}
Away from the tree limit, one needs to terminate the algorithm after a certain number of steps. The number of steps is limited by the size of the cluster constructed, and one can compute $d$ coefficients from the cluster of radius $d$. Using the largest accessible value of $b_d$ and assuming $b_{n\geq d}=b_d$ one terminates the continued fraction by replacing the rest with a similar closed form expression~\cite{haydock_1972}. This procedure is described in~\cite{mosseri_2023} as embedding the cluster into an effective tree-like medium to suppresses finite-size effects associated with the large fraction of sites at the boundary. 

Given the approximation to $G(z)$ we calculate
\begin{equation}
    \mathcal{D}(\epsilon) = \bigg | \frac{1}{\pi}\  \text{Im}\ G(\epsilon + i0^+) \bigg |
\end{equation}
which corresponds to the \emph{average} RMT many-body density of states of the original spin chain via the RMT/single-particle duality. Results are shown in Fig.~\ref{fig:density_of_states} compared against exact diagonalization of a single realization of the corresponding spin chain. The purple histograms measure a coarse-grained or smeared density of states via
\begin{equation}
    \mathcal{D}(\epsilon) = q^{-L}\sum_{n} \delta(\epsilon-E_n)
\end{equation}
where $E_n$ are the many-body eigenenergies of a single random realization of $H$ and $\delta(x)$ is a small window function centered at $x$, e.g. a rectangular window. It is clear from this numerical result that the density of states is a sufficiently coarse gained quantity that it is self averaging, i.e. a single realization is already the same as the average. 

\begin{figure}[b]
    \centering
    \includegraphics[width=0.6\linewidth]{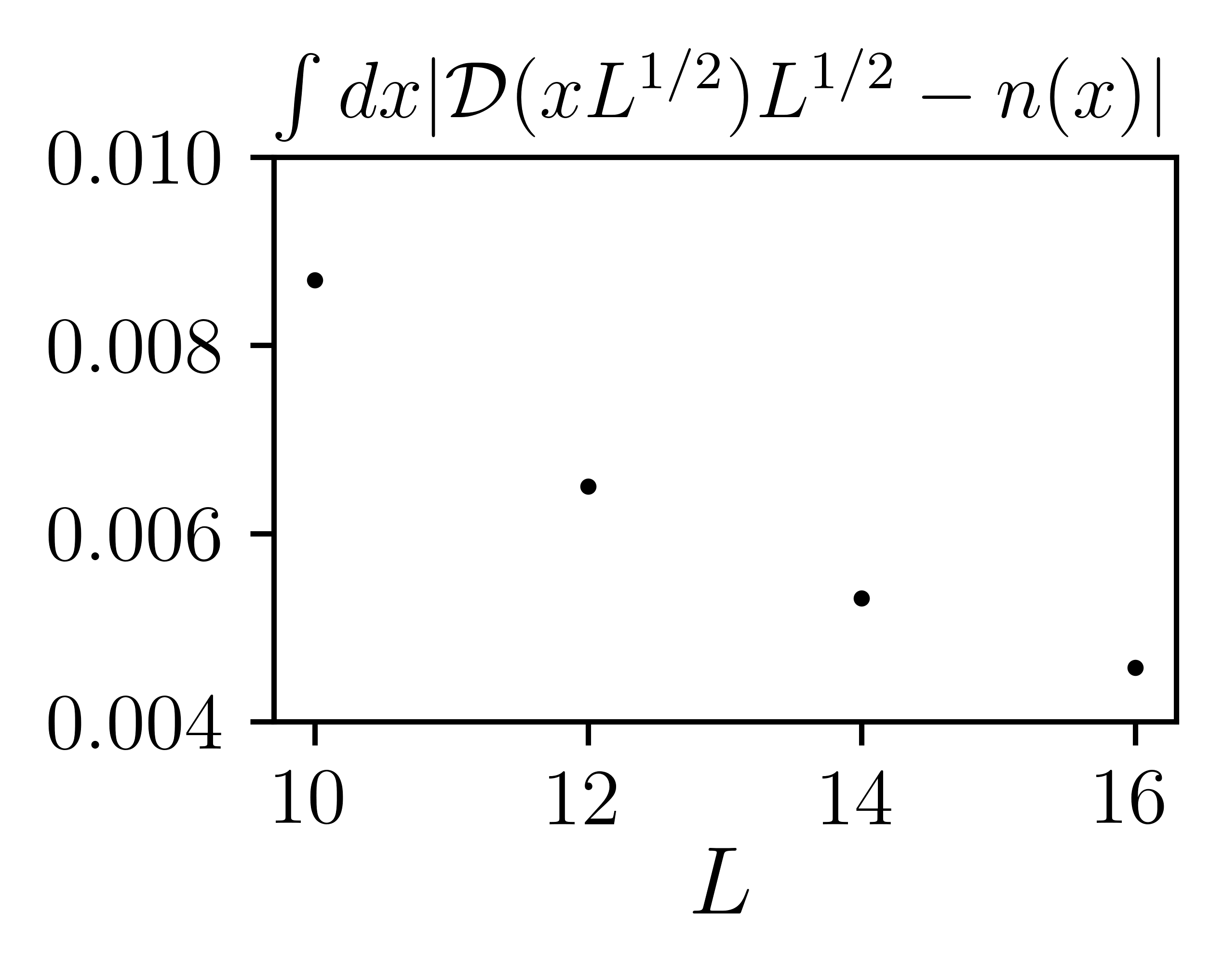}
    \caption{The properly scaled density of states obtained from the group word hopping problem approaches a normal distribution as $L$ grows.}
    \label{fig:compare_gaussian}
\end{figure}

An important empirical observation is that the density of states evolves towards a Gaussian as $L\rightarrow \infty$. While this is the physical result one would anticipate for a locally interacting many-body system, it obtained here in spite of first having taken $q\rightarrow\infty$. For any finite $L$, the $q\rightarrow\infty$ limit leads to a sharp band edge with no eigenvalues outside. Besides, the maximum possible eigenvalue of $H$ for any $q$ is strictly upper bounded by $L$ because the bonds have eigenvalues $\pm 1$. However, as $L$ is increased, most of the eigenvalues concentrate within $|\epsilon| \sim \sqrt{L}$ and the spectrum develops approximately Gaussian tails such that a large eigenvalue $\sim aL$ is exceedingly unlikely. To check the emergence of the Gaussian as $L$ grows, we can approximately compute the total deviation of the properly scaled density of states from the standard normal $n(x) = (2\pi)^{-1/2}e^{-x^2/2}$. Fig.~\ref{fig:compare_gaussian} shows that this error is decreasing with $L$. This is completely distinct from the semicircular law of width $\sqrt{L}$ which would be approached if we considered the sum of $L$ mutually free variables as $L$ grows~\cite{nica_lectures_2006}.

The $L=5$ spin chain displays a particularly fast convergence of the sequence $b_n$, while the convergence is slower for $L=6$. The $L=4$ result is shown in Appendix~\ref{sec:dos_details}, which has particularly poor convergence properties due to a sharp peak (likely a singularity) at $\epsilon=0$ which were also noted in Ref.~\cite{pollock_2025}. Slow convergence of $b_n$ is known to occur in the presence of nonanalyticities~\cite{mosseri_2023}. However for $L\geq 5$, there do not appear to be any nonanalyticities in the bulk of the spectrum. We remark that as $L$ grows, we can get away with iterating the Lanczos algorithm fewer times, as the resulting error becomes decreasingly important on a scale growing with $L$.

For the interested reader, in Appendix~\ref{sec:dos_details} we collect data for the moments $\braket{H^n}$, which are the number of walks returning to the origin and also called the co-growth series. The number of points on the Cayley graph within a radius $d$ of $e$, also called the growth function or volume, $V(d)$, and all of the collected data for the Lanczos coefficients $b_n$ are also collected for reference. We note that explicit analytical computation of the above quantities, especially the co-growth, has been so far elusive in the literature on such right-angled Coxeter groups, but relatively recent work has discussed the asymptotic behavior of the growth series \cite{kolpakov2020}, and Appendix~\ref{sec:dos_details} may be of interest to researchers in that area.

\subsection{Two-point functions}

\begin{figure}
    \centering
    \includegraphics[width=0.9\linewidth]{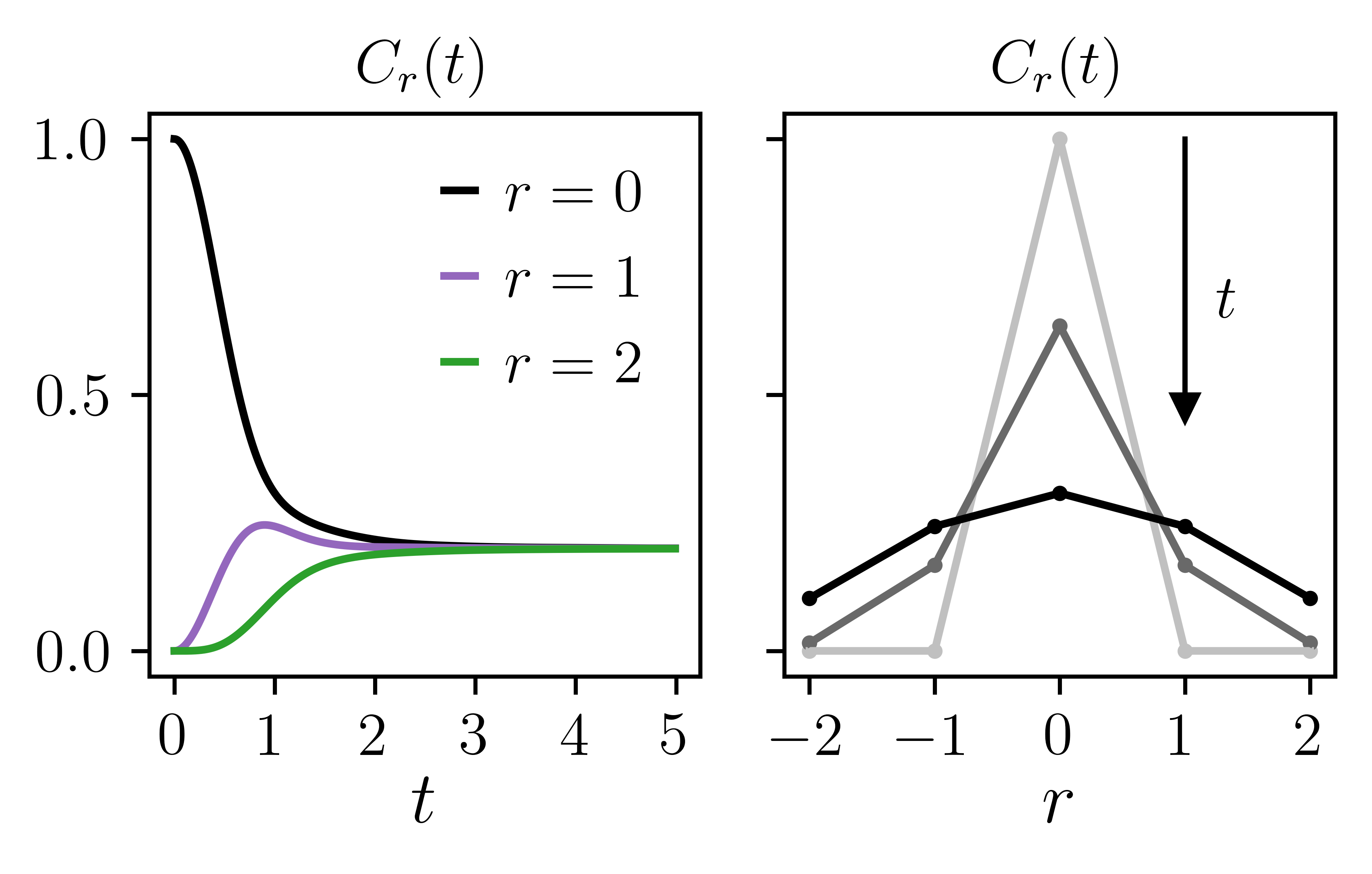}
    \caption{Dynamics of the local energy density in the infinite-$q$ and $L=5$ local HI RMT chain. $C_r(t)$ is the response at distance $r$ and time $t$ following a local excitation on one bond. Data is obtained numerically via single-particle hopping dynamics on the Cayley graph of $G_5$ including points up to radius $d=16$.}
    \label{fig:two_point}
\end{figure}

The next use of the correspondence is to calculate correlators of the local energy density. In particular, we consider two point functions of the local energy density
\begin{equation}
    C_r(t) = \braket{h_i(t)h_j(0)} \quad |i-j|=r
\end{equation}
which are, physically, proportional to the high temperature response at a distance $r$ and time $t$ to an energy impulse on some bond~\footnote{For the HI model with $h_i^2=1$, these are not only the linear response, but the full non-perturbative response.}. Here the average two-point functions depend only on the separation $r$ since the system is translation invariant on average.

For these calculations, we simply directly use finite-size approximations to the $V_i$ including points up to radius $d$, and recognize that after a certain time-scale the results will be invalid, which can be made self-consistent by comparing the calculation for say, radius $d-1$, $d$, and $d+1$, and ignoring the result after they differ by some small value. Within this approach, we are limited to smaller $L$, as it becomes impractical to calculate the deep adjacency matrices for the large $L$ graphs with many generators. Since the $L=5$ case has particularly nice convergence properties, we focus on this case, but we show more available results for $L=4,6$ two-point functions in Appendix~\ref{sec:dos_details}.

While the physics behind these two point functions is relatively simple, they are highly non-trivial from a random-matrix point of view. Should one have considered an all-to-all $H$, such spatially extended correlation functions would not be meaningful. It is interesting to see that the results obtained from the mapping to a free particle are indeed physical. The energy impulse thermalizes on an $O(1)$ timescale as $q\rightarrow\infty$ despite the fact that the Heisenberg time is formally infinite. It also takes time for energy to spread through the chain, for example distant ($r\geq 1$) correlators are initially small since the energy density hasn't reached further bonds. The latter can be made more precise, as follows.

In Fig.~\ref{fig:lieb_robinson}, we show maximally distant two-point functions at early time for various system sizes $L$. We are able to access larger $L$ here since we are only interested in early times, and, consequently, smaller Cayley graphs suffice for accurate results. One can see an exponential decrease of correlatons with $L$ for any fixed $t$. This is simply the statement that there exists an emergent causality, i.e. that any signal at large separation and short time should be exponentially suppressed. This is already a signature of locality of interactions. One can also see a power-law increase with time that is consistent with perturbation theory. The two-point function can be expanded as
\begin{equation}
    C_r(t) = \sum_n \frac{(it)^n}{n!} \sum_{i_1 \cdots i_n} \braket{[h_{i_1}, \cdots,[h_{i_n},h_{r+1}]\cdots] h_1}.
\end{equation}
The leading behavior at early times will be given by $t^n$ for the smallest $n$ such that the both the commutator and the RMT correlator are non-vanishing. For example, supposing $r=3$, a shortest-length word satisfying the above conditions is
\begin{equation}
    \braket{h_1h_2h_2h_4h_3h_3h_4h_1}
\end{equation}
In general we need $n=2r$, leading to the early-time dependence $\sim t^L$ for $r=L/2$ seen in Fig.~\ref{fig:lieb_robinson}.

\begin{figure}
    \centering
    \includegraphics[width=0.62\linewidth]{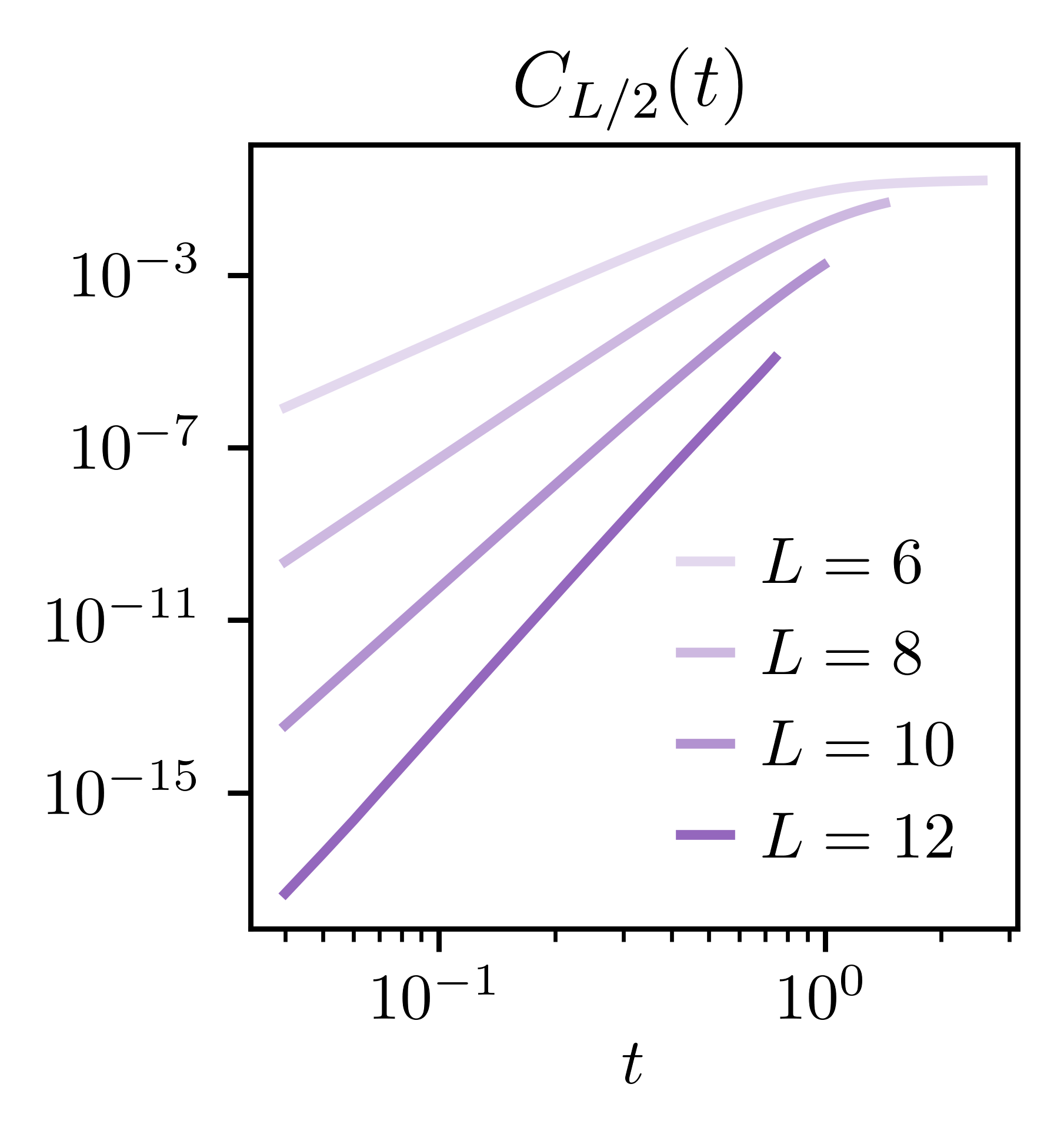}
    \caption{Infinite $q$ early time behavior of maximally distant two-point functions ($r=L/2$) for each system size $L$, computed via the group word hopping. More distant correlators take longer to receive a significant signal due to the locality of interaction and the early time dependence is $\sim t^L$.}
    \label{fig:lieb_robinson}
\end{figure}

Fig.~\ref{fig:two_point} panel (a) shows the spacetime dependence of the energy density computed using the Cayley graph of $G_5$ up to radius $d=16$. On the plotted timescale, $d=15,16$ were indistinguishable confirming the validity of the truncation. It is interesting to see that the late-time value of $C_r(t) \rightarrow L^{-1}$, without further finite-size corrections. The fact that the long-time value is precisely $L^{-1}$ for small system, including for $L=2,3$ (which we observed in Ref.~\cite{pollock_2025}) without further corrections, is an interesting result. Unfortunately, the previous and present methods confirm only that the $L=2,3,4,5$ two-point functions equilibrate to $L^{-1}$, since $L\geq 6$ calculations become impractical to obtain large enough adjacency matrices to achieve sufficiently late times. Nonetheless, we conjecture that this holds for any $L$.

\subsection{Typicality and Mazur bounds}

\begin{figure}
    \centering
    \includegraphics[width=1\linewidth]{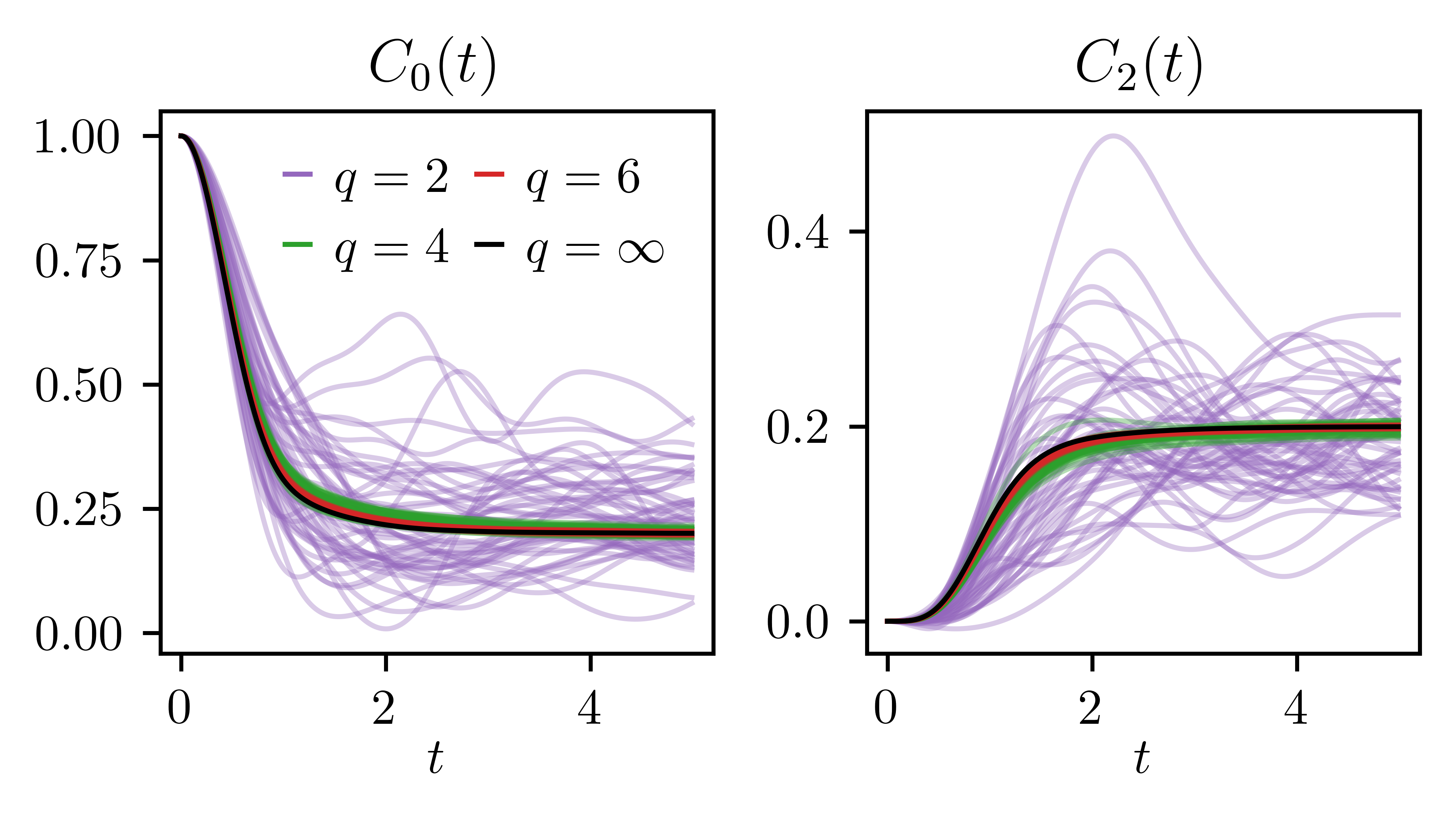}
    \caption{Examining the typicality of two-point functions for $50$ different instances of the $L=5$ HI RMT chain. Each line is a unique realization for $q=2,4,6$ in purple, green, and red, respectively. These are compared against the $q\rightarrow\infty$ result in black. The $q=2$ chain is clearly not self-averaging, while already $q=6$ is highly self-averaging.}
    \label{fig:typicality}
\end{figure}

In Sec.~\ref{sec:dos} where we discussed density of states, it was empirically clear that such a quantity is self-averaging. Here, we ask if the two-point functions are also self-averaging as $q\rightarrow\infty$ for finite $t$ and $L$. We again focus on $L=5$ since this is the most nontrivial model we can still access at long times with the present methods. In Fig.~\ref{fig:typicality} we compare $50$ different realizations of the $L=5$ HI RMT chain for small $q$ to the $q=\infty$ average result. One can see that $q=2$ is not self averaging. However, if we increase to $q=4$ it is clear that already the $q\to\infty$ physics is relevant. For $q=6$, there is not a single outlying curve among the $50$ samples.

Let us digress on the typicality of the equilibrium value, in particular for the autocorrelator $C_0(t)$. We have seen in the cases studied thus far with $L\leq 5$ that this approaches $L^{-1}$. That value occurs in general for a translation-invariant system whose only integrals of motion that couple to energy density are powers of the Hamiltonian $H^k$. This can be seen from Mazur bounds on the autocorrelation function, which we discuss in Appendix~\ref{sec:mazur}. In the HI RMT chain, there are some subtleties related to disorder averaging. We provide some steps towards a proof that the long-time value of $C_0(t)$ is $L^{-1}$ by showing that, if certain correlations are unimportant in the limit $q\rightarrow\infty$ and if the only integrals of motion coupling to energy density are powers of $H$, we obtain the Mazur ``equality" $C_0(t) = L^{-1}$. The equality means that the bound is saturated since we have putatively accounted for all charges. We conjecture that for the group $G_L$ and some element $Q$ of the group algebra,
\begin{equation}
    [Q,\Delta] = 0 \implies Q \in \text{alg}(e,\Delta)
\end{equation}
where $\text{alg}(e,\Delta)$ is the algebra generated by $e,\Delta$ and that some similar statement holds in the original spin chain with high probability as $q\rightarrow\infty$.

\section{From chaos to integrability}\label{sec:integrable}

In the previous sections, we discussed the HI random-matrix spin chain, which is quantum chaotic in the sense of having Wigner-Dyson energy level spacing statistics in both the large $q$/small $L$ and small $L$/large $q$ limits~\cite{pollock_2025}. We also conjectured that the only local conservation law was total energy in the large $q$ limit. The emergent group word dynamics for that model occurs on an infinite right-angled Coxeter group. Since the only relations in such a group are that distant generators commute, for any given $L$ this group is the ``biggest" we can have subject to $g_i^2=e$ and the geometry of the spin chain. The key property is that neighboring generators are free, and so
\begin{equation}
    (g_i g_{i+1})^\infty = e ,
\end{equation}
that is, no finite length alternating group word of the above form can ever return to the identity in the group. For the above reasons, we refer to the large $q$ model as chaotic. What happens if we interrupt this freeness and make the group ``smaller," such that the neighboring generators obey
\begin{equation}
    (g_i g_{i+1})^k = e
\end{equation}
for some finite $k$? We might expect that the dynamics of the system will be more special and less generic or less chaotic.

In this section, we discuss in some detail what would happen if we set $k=3$. This corresponds to putting \emph{braids} (i.e., Yang-Baxter relations) into the group. Upon constructing two more nontrivial local conserved charges beyond energy, we conjecture that adding braids is sufficient to render the chain integrable at the level of the group algebra. At the end of the section, we briefly touch upon setting $k=4$ and leave the rest for future work. We will compare these various cases at the level of an abstract group which may or may not admit a ``matrix model", however, we do point out a deterministic translation invariant matrix model when we include braids and an additional ``long" relation in the group which yields the symmetric group. Besides this case, we will refer generally only to the group generators $g_i$ and not actual matrices $h_i$ acting on some spin chain. The groups we study in this section are all translation invariant in the sense that we can rotate all generators $g_i \rightarrow g_{i+a}$ and obtain the same group. This is also true for the group $G_L$ corresponding to the RMT chain, which was not strictly translation invariant as a spin chain.

\begin{figure}
    \centering
    \includegraphics[width=\linewidth]{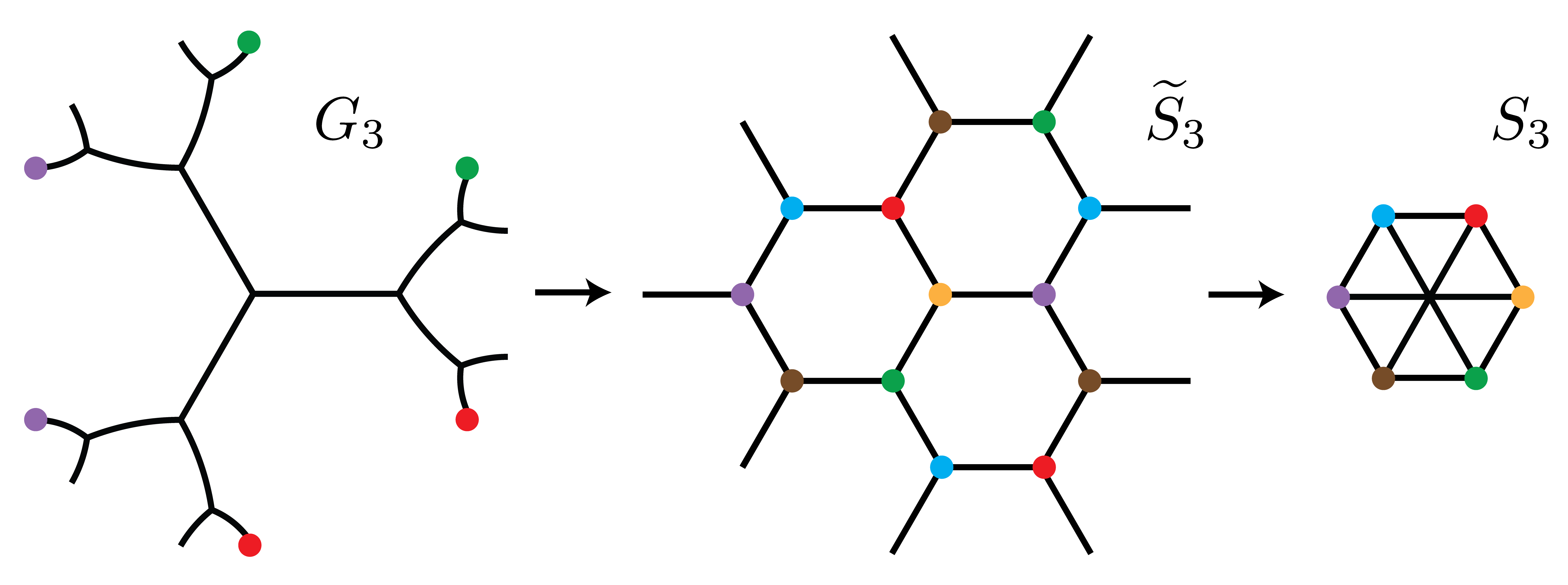}
    \caption{How the Cayley graph changes as we go from the right-angled Coxeter group down to the symmetric group. The graph of $G_3$ is an infinte tree, the graph of $\widetilde{S}_3$ is an infinite honeycomb lattice, and the graph of $S_3$ a small finite graph with $3!=6$ nodes. Colored circles represent points that are identified as the appropriate relations are added. Note that the symmetric group with the given presentation does not form a planar Cayley graph and the central point where all lines cross is not a group element.}
    \label{fig:collapsing_3}
\end{figure}

\subsection{Braids and the symmetric group}

\begin{figure*}
    \centering
    \includegraphics[width=0.85\linewidth]{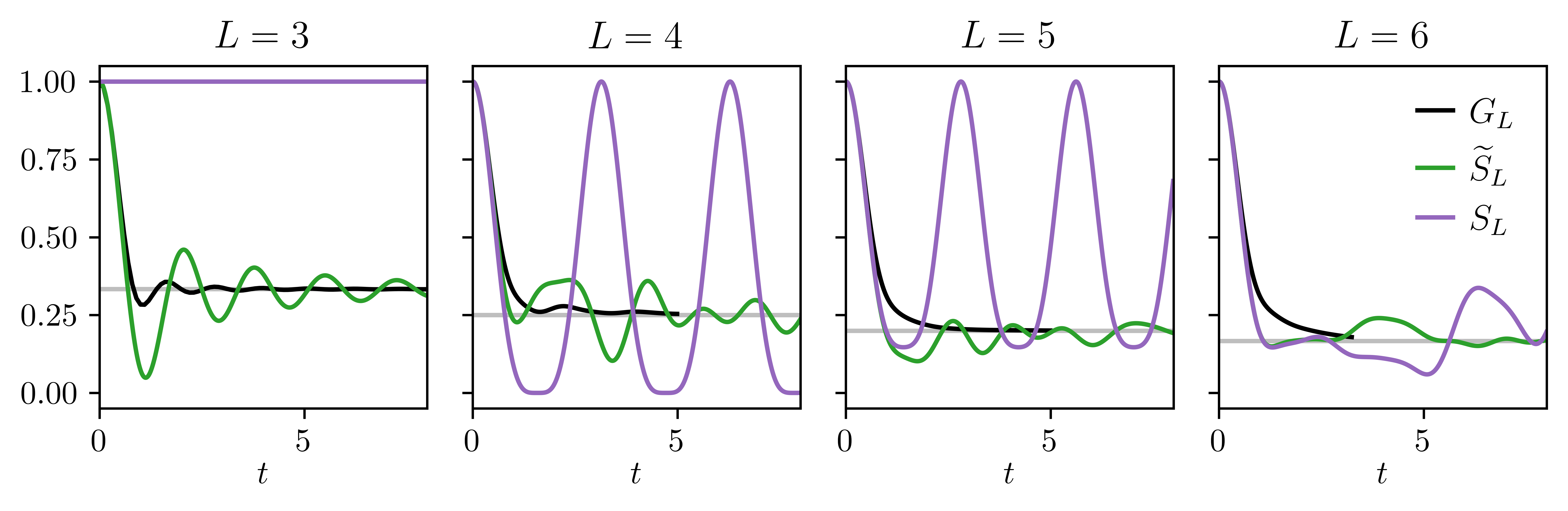}
    \caption{How autocorrelation functions $C_0(t)$ change as we add braid relations and the ``long" relation Eq.~\eqref{eq:long} into the group. The dynamics of $G_L$ (black curves) is the most ``generic" in the sense that the correlators quickly reach $L^{-1}$ (shown in light gray horizontal lines). Adding the braids brings one to the group $\widetilde{S}_L$ which displays some strong oscillations, while going all the way to the symmetric group leads to some interesting revivals for $L=4,5$, though not in general.}
    \label{fig:modify_group}
\end{figure*}

Adding braid relations into the group $G_L$ gives the affine symmetric group $\widetilde{S}_L$, which has the presentation
\begin{equation}
    \widetilde{S}_L = \braket{\ g_j\  |\ g_j^2,\ (g_i g_{i+1})^3,\  (g_i g_j)^2\ \text{if}\ |i-j|>1\ } 
\end{equation}
with $L$ generators. This is still an infinite group, but its word growth is no longer exponential. In fact, the geometry of the Cayley graph with respect to this presentation is Euclidean for any $L$~\cite{lewis_affine_2021}.
Just looking at this presentation, we see it is almost the symmetric group $S_L$, as generated by transpositions of adjacent symbols
\begin{equation}
    g_i = (i, i+1) ,
\end{equation}
which braid, square to $e$, and commute if far enough apart. However, there is no rule giving that, e.g. in $L=3$, $g_1 g_2 g_1 \sim (12)(23)(12) = (13) \sim g_3$. As a result, the group is infinite, whereas the symmetric group $S_L$ is finite. $\widetilde{S}_L$ acts on the Cayley graph in the same way as discussed in Sec.~\ref{sec:tools}: draw perpendicular hyperplanes to the edges and reflect. These reflections can be used to define $\tilde{S}_L$ as the reflection symmetries of an $L-1$-dimensional hyperplane in $L$-dimensional affine Euclidean space.

To go all the way down to the symmetric group, we need one more ``long" relation to impose this condition,
\begin{equation}\label{eq:long}
    g_L = g_1 g_2 \cdots g_{L-2} g_{L-1} g_{L-2} \cdots g_2 g_1.
\end{equation}
We note that adding this relation makes the presentation non-Coxeter, but the automatic structure remains. Upon adding this relation, we can now understand the group word dynamics in $S_L$ as the energy dynamics of a PBC large-$q$ chain with SWAP interactions:
\begin{equation}\label{eq:swap}
    H = \sum_{i=1}^L h_i,\quad h_i = \sum_{a,b=1}^{q} \ket{ba} \bra{ab} \otimes I_{\overline{i, i+1}}.
\end{equation}
This follows because, much like the HI chain, we have that in the $q\rightarrow\infty$ limit
\begin{equation}
    \braket{h_{i_1} \cdots h_{i_k}} = 
    \begin{cases}
        1\quad \mathrm{if}\ \pi = e\\
        0 \quad \mathrm{otherwise}
    \end{cases}
\end{equation}
where $\pi = g_{i_1}\cdots g_{i_k}\in S_L$ is the permutation obtained by composing all of the SWAPs in the group word (the action can be read from the right or from the left) \footnote{At finite $q$, we have
\begin{equation}
    \braket{h_{i_1} \cdots h_{i_k}} = q^{\#(\pi) - L}
\end{equation}
with $\#(\pi)$ the number of cycles in the permutation $\pi$ (including trivial cycles).}. We can therefore make all the same identifications Eqs.~\eqref{eq:correspondences} as we did in the RMT case; this model still admits a single-particle group word dynamics, but now the group is finite.

How do the Cayley graphs and the energy dynamics vary as we go through the sequence
\begin{equation}\label{eq:sequence}
    G_L \longrightarrow \widetilde{S}_L \longrightarrow S_L?
\end{equation}
The easiest case to understand is $L=3$, since the Cayley graphs of both $G_3$ and $\widetilde{S}_3$ are planar, and one can clearly see how the geometry changes: see Fig.~\ref{fig:collapsing_3}. By adding the braids, which contain $6$ generators, we change from the hyperbolic tiling $\{\infty,3\}$ to a Euclidean one, $\{6,3\}$---the honeycomb lattice. The group is still infinite, but is in a sense smaller: there are fewer words of each minimum length greater than 2. Going all the way to the symmetric group, we get a finite Cayley graph. Let us now look at how the energy dynamics following an excitation on some bond differs in these different groups.

We focus on autocorrelators $C_0(t)$ for simplicity. In this section, we mean the abstract correlator
\begin{equation}
    C_0(t) = \braket{g_i(t) g_i(0)}_e ,
\end{equation}
where $\braket{\cdot}_e$ means pick out the term proportional to $e$ in the group algebra, and Heisenberg time evolution refers to single-particle hopping \footnote{We have slightly abused notation in this section. Previously, the adjacency matrix $\Delta$ acted on the single-particle Hilbert space of the Cayley graph, whereas here $\Delta$ is as an abstract element of the group algebra,
\begin{equation}
    \Delta \equiv \sum_i g_i.
\end{equation} Equivalently, we can replace all $g_i$ appearing in the expressions with the matrix representations $V_i$, again acting on the single-particle Hilbert space of the Cayley graph and get the same result.}. Fig.~\ref{fig:modify_group} shows the autocorrelator for $L=3,4,5,6$ and the different groups. We can see in general that going through the sequence defined by Eq.~\eqref{eq:sequence} makes the dynamics more special. The most generic situation is when there is no relation at all among $g_i$ and $g_{i+1}$ due to the asymptotic freeness in the original RMT chain. Those curves quickly (compared to the other groups) thermalize to $L^{-1}$. However, upon adding braids, pronounced oscillations are observed. Going all the way to the symmetric group, i.e. the dynamics under the large $q$ SWAP chain Eq.~\eqref{eq:swap}, we can see that $L=3,4,5$ are exceptional.

For $L=3$, adding all of the relations renders the dynamics a bit too non-generic: no dynamics occurs at all. In fact, one can check that, in that case, $[g_i,\Delta]=0$ in $S_3$. For $L=4$, perfect constructive and destructive interference of energy density is observed at some periodic times $t_*$, for which $[g_i(t_*),g_i]=0$. One can see this by examining the OTOC, which is essentially the norm of the commutator. For $L=5$, the same holds for perfect revivals, but no perfect destructive interference is observed. For $L\geq 6$ no revivals are observed. As $L$ grows, the Cayley graph becomes complicated, so many frequencies contribute to the dynamics and full revivals are absent. However, we observed that $S_L$ for $L=6,7,8$ displays persistent but irregular fluctuations which do not decay, consistent with the group being finite.

\subsection{More local charges beyond energy}

In this section we focus on the group $\widetilde{S}_L$. Everything that follows will apply to $S_L$ as well, but braid relations (without interpreting the generators as SWAPs) and self-inverses are sufficient for the construction of higher translation-invariant charges. The braid relations
\begin{equation}
    (g_i g_{i+1})^3 = e
\end{equation}
are reminiscent of the Yang-Baxter equation, but in a simplified, parameterless setting. Yang-Baxter integrability implies an infinite set (for an infinite chain) of conserved local translation invariant conserved charges all mutually in involution and linearly independent~\cite{faddeev_how_1996}. We ask if something similar happens here.

One approach to generate these charges uses the so-called boost operator, which acts as a ladder operator for the charges~\cite{loebbert_2016}. While we do not intend to study this topic in detail, we do show that the boost operator can generate at least two more charges beyond the energy itself. We show this purely at the level of the group algebra (linear combinations of group elements) without reference to any particular spin chain. To generate higher charges, we start with the ``energy" or hopping adjacency matrix itself:
\begin{equation}
    Q_2 = \sum_{i=1}^L g_i.
\end{equation}
Arguably this should be called $Q_1$, since it involves radius one group words, but we start at $Q_2$ to be consistent with the literature.
The boost operator is defined by
\begin{equation}
    B = \sum_{j=1}^L j\ g_j.
\end{equation}
Let
\begin{equation}
    Q_3 = [B,Q_2] = \sum_{i=1}^L [g_i,g_{i+1}].
\end{equation}
We claim that for a sufficiently large $L$, this is both linearly independent from and in involution with $Q_2$. An important thing one must do when generating the charges in this way on a finite chain of length $L$ with PBC is to read the coefficients of $[B,Q_n]$ modulo $L$. Otherwise, some boundary terms will arise beyond what is written. We note that $Q_3$ is proportional to an energy current: $J= i Q_3$ and $J_j = i [g_j,g_{j+1}]$ is the current density operator~\cite{Bertini2021}. Therefore, braids alone are sufficient to produce a persistent energy current. It was recently proven in Ref.~\cite{hokkyo_integrability_2025} that, for a class of spin chains, existence of this persistent energy current is enough to guarantee the whole set of charges. Since we are operating at the level of an abstract group algebra, we cannot directly import this result.

We can, however, use the boost operator to generate another charge:
\begin{equation}
    Q_4 = [B,Q_3] - 2Q_2 = \sum_{i=1}^{L}\ [[g_i, g_{i+1}], g_{i+1}+2g_{i+2}]
\end{equation}
where we subtracted off a trivial term proportional to $Q_2$. This charge is in involution with both $Q_2$ and $Q_3$. One can verify that, in $\widetilde{S}_L$, say for $L\geq 5$,
\begin{equation}
    [Q_3,Q_2] = 0 \quad [Q_4,Q_2] = 0 \quad [Q_4,Q_3] = 0.
\end{equation}
Because the calculations become extremely tedious, we outline a proof in particular that $[Q_4,Q_3] = 0$ in Appendix~\ref{sec:charges_in_involution}. The other equalities are verified by a similar method. The method is essentially to make repeated use of the braiding, the fact that the generators square to $e$, and translation invariance. 

We can also see that the charges are linearly independent for sufficiently large $L$, since they involve increasingly long group words. Via repeated application of the boost operator, charge $Q_n$ will contain a component proportional to 
\cite{hokkyo_integrability_2025}
\begin{equation}\label{eq:nesty_boi}
    D_n = \sum_{i=1}^L [\dots[[g_i,g_{i+1}],g_{i+2}]\dots, g_{i+n-2}].
\end{equation}
Let us first consider the group $\widetilde{S}_L$. We note that $D_n$ is then a linear combination of $2^{n-2}\cdot L$ distinct group elements, provided that $n\leq L$. Therefore no terms cancel and $D_n$ is a genuine nonzero contribution to $Q_n$ for $n\leq L$. Furthermore, the group elements in this sum all have word radius $n-1$: they cannot be reduced by braid relations because only one of each generator appears in each term. One may then wonder if in the symmetric group $S_L$ anything is different. If fact, the ``long" relation Eq.~\eqref{eq:long} and its translates do not enter, so long as we consider words of radius at most $L-1$. Therefore in both groups, for $n\leq L$, the quantities $D_n$, and consequently $Q_n$, are linearly independent from $Q_{n-1}$, which contains at most radius $n-2$ words.


We therefore conjecture that for increasing $L$, we can generate extensively many linearly independent local charges mutually in involution using the boost operator, i.e. that the system is integrable.

\begin{figure}
    \centering
    \includegraphics[width=\linewidth]{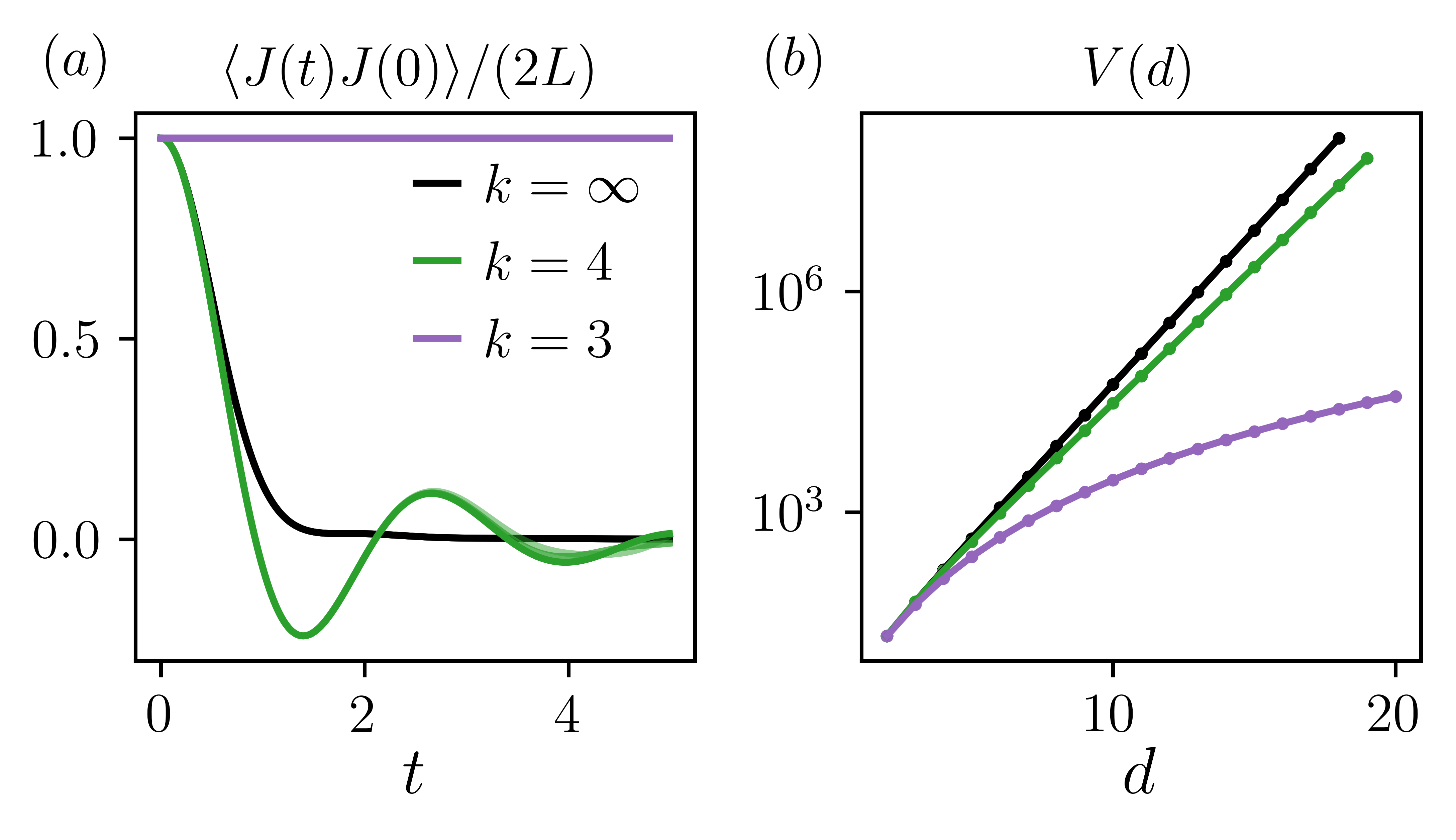}
    \caption{In (a), the energy current-current correlator for an $L=5$ PBC chain computed from the group word dynamics for three different infinite groups. The index $k$ refers to the relation between neighboring generators $(g_ig_{i+1})^k=e$ in the group. The three deviating green curves for $k=4$ are computed for Cayley graphs up to radius $d=11,12,13$ (shown in three shades of green) and so one can see that they start to deviate, but the general trend of oscillations is clear. In (b), the growth or volume of the group. While $k=3$ grows polynomially, both $k=4$ and $k=\infty$ grow exponentially.}
    \label{fig:energy_current}
\end{figure}

\subsection{Beyond braiding and the word growth}

Here we wrap up our discussion of modifying the relation between neighboring generators in the group. We ask if, like $k=3$, anything special can obtain upon setting $k=4$. At the outset, it seems unlikely that there will be a persistent energy current, since the braids were crucial to obtain $[Q_3,Q_2]=0$. One might hope, however, that this decays more slowly. To check this, we compare how a small injected energy current in an $L=5$ chain decays for $k=3,4,\infty$ in Fig.~\ref{fig:energy_current}. We see that while there is a perfectly conserved energy current in the $k=3$ instance, already at $k=4$ the energy current immediately decays in the same manner as $k=\infty$. However, unlike the free case, it does appear to oscillate around zero; the current sloshes back and forth before equilibrating (assuming it does equilibrate to zero). The amplitude also appears to decay more slowly than for $k=\infty$. As was a theme in the previous section, the dynamics is less ``generic" than the free case, but still almost certainly non-integrable.

It seems likely that there is a relationship between integrability and the ``size" of the group which determines the energy dynamics. This is captured by the scaling of the number of group elements within graph distance $d$ of $e$, also known as the ``volume" or ``growth function" $V(d)$.  For $k=\infty$, the group is a right-angled Coxeter group and so the volume or growth function $V(d)$ of the group is known to be strictly exponential (for $L\geq 3$). For the $L=5$ case for example, $V(d) \sim \varphi^{2d}$ as $d\rightarrow \infty$ with $\varphi$ the golden ratio (the exact sequence is entry A122678 in the Online Encyclopedia of Integer Sequences). More generally, the growth functions of right-angled Coxeter groups are known to be asymptotically $p^d$ with $p$ a Perron number \cite{kolpakov2020}. 

While the growth function of the affine Symmetric group $\widetilde{S}_L$ is polynomial, specifically $V(d) \sim d^{L-1}$, we can see in Fig.~\ref{fig:energy_current} that already the $k=4$ group (which is a Coxeter group but not a right-angled one) already has an exponential growth function. This suggests that the $k=3$ group is exceptional, and that an extensive number of local integrals of motion seem unlikely to exist for larger $k$. Nonetheless, one may wonder if smaller $k$ than infinity can host slower dynamics and if this is related to the scaling of $V(d)$, a question we leave for future work.

\section{Conclusion and outlook}\label{sec:conclude}

In this work we have continued our study, initiated in~\cite{pollock_2025}, of a certain random matrix ensemble which of $L$ locally interacting spins on a PBC chain. Using results from free probability theory, we have shown that any infinite temperature correlation function of only energy density operators in this random matrix model can be cast as a single-particle hopping process on an infinite regular graph in the limit of large local Hilbert space dimension. The graph is the Cayley graph of a right-angled Coxeter group. Tools from computational group theory allowed us to compute huge adjacency matrices on this graph and obtain approximate numerical solutions which revealed the physics of a generic Hamiltonian dynamics consistent with locality, including approach to a Gaussian density of states and spreading and thermalization of energy density.

The emergent group word dynamics that we established in the random matrix chain can then be extended to a more general setting. As an example, we studied the case of adding braid relations into the group, and saw that this implies more local conserved charges, including energy current. This led us to conjecture that the system is integrable at the level of the group algebra. Adding one more relation turns the group into the symmetric group generated by adjacent SWAPs, for which a microscopic Hamiltonian can be written which realizes the corresponding group word dynamics, again in the limit of large local Hilbert space dimension. In that case exact revivals of energy density can sometimes be observed.

As a final example, we discussed another Coxeter group with the relations $(g_i g_{i+1})^4 = e$, finding that this group has a slower decaying energy current than the RMT chain with relations $(g_i g_{i+1})^\infty = e$, but not a persistent energy current, as in the case with $(g_i g_{i+1})^3 = e$, suggesting that braids are special. It seems possible that one could generate a hierarchy of quantum chaotic dynamics by tuning $k$, similar in spirit to $k$-designs~\cite{roberts_chaos_2017}.

\begin{figure}
    \centering
    \includegraphics[width=0.6\linewidth]{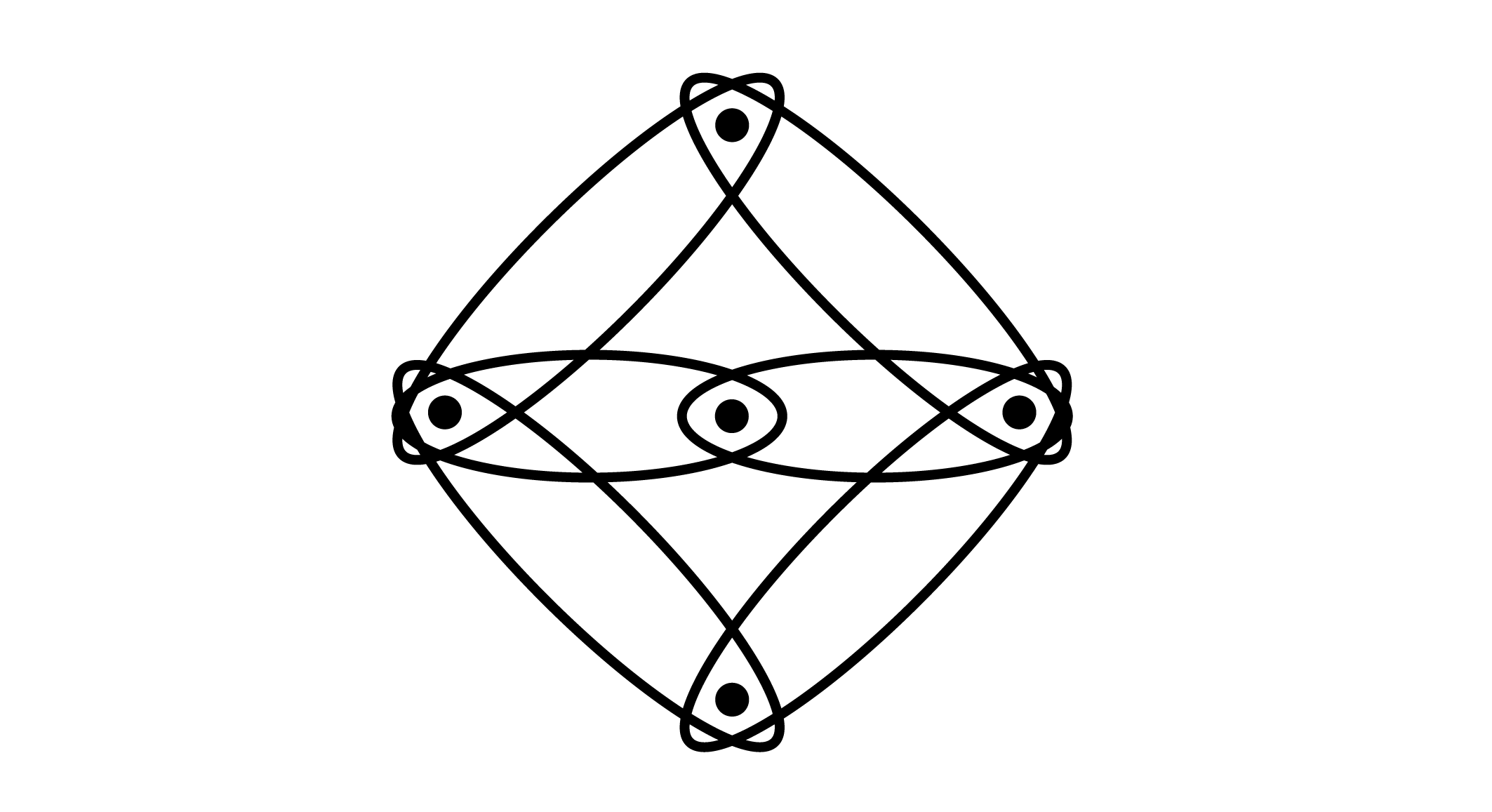}
    \caption{A ``matrix model" made of Haar-Ising bonds which realizes one more Coxeter group not so far discussed: the hyperbolic $\{4,6\}$ tiling.}
    \label{fig:4_6}
\end{figure}

One direction going forward is to understand the nature of the conserved quantities at the level of the group algebra in these various models, and what that means for the putative emergent hydrodynamics of energy density. For the groups $G_L$, we conjectured that the only integrals of motion are trivial, i.e. living in $\text{alg}(e,\Delta)$. This implies a very strong non-integrability of the HI chain. One would therefore expect energy transport to be diffusive. In Ref.~\cite{pollock_2025}, we found that $q=2$ chains up to length $L=20$ were not clearly diffusive, but we believe that the transport should be purely diffusive for large $q$. To prove this, we likely need to get a handle on the emergent geometry of the hopping problem for large $L$ and see how a diffusive scaling $t^{-1/2}$ for two-point functions arises at the level of the group word dynamics. A related question, in the case where braids are added, is whether one could prove Yang-Baxter integrability and see how a ballistic scaling $t^{-1}$ of two-point functions arises.

Another direction is to study further the correspondence between Coxeter groups in a given presentation and local Hamiltonians. One can turn the RMT single particle duality around; given a Coxeter group, can we take advantage of asymptotic freeness of random matrices to build a ``matrix model" which realizes the desired hopping problem? For hyperbolic tilings, beyond the cases $G_{3,4,5}$ discussed in the text, we give the example of a $\{4,6\}$ tiling in Fig.~\ref{fig:4_6} using HI random matrices arranged in a certain geometry.

\acknowledgements

The authors are grateful for discussions with Benedikt Placke, Tathagata Basak, Jonas Hartwig, Rustem Sharipov, Anushya Chandran, Marko \v{Z}nidari\v{c}, Toma\v{z} Prosen, and especially with Alexey Khudorozhkov who pointed us to the literature on Coxeter groups. We thank R\'emi Mosseri for sharing some data with us. K.P. and T.I. acknowledge support from the National Science Foundation under Grant No. DMR-2143635. J.D.K. acknowledges support from the U.S. Department of Energy under Grant No. DE-SC0023692. J. R. acknowledges financial support from the Royal Society through the University Research Fellowship No. 201101.

\begin{appendix}

\section{\texorpdfstring{$G_L$}{GL} word growth, moments, convergence of Lanczos coefficients, and two-point functions}\label{sec:dos_details}

\begin{figure*}
    \centering
    \includegraphics[width=\linewidth]{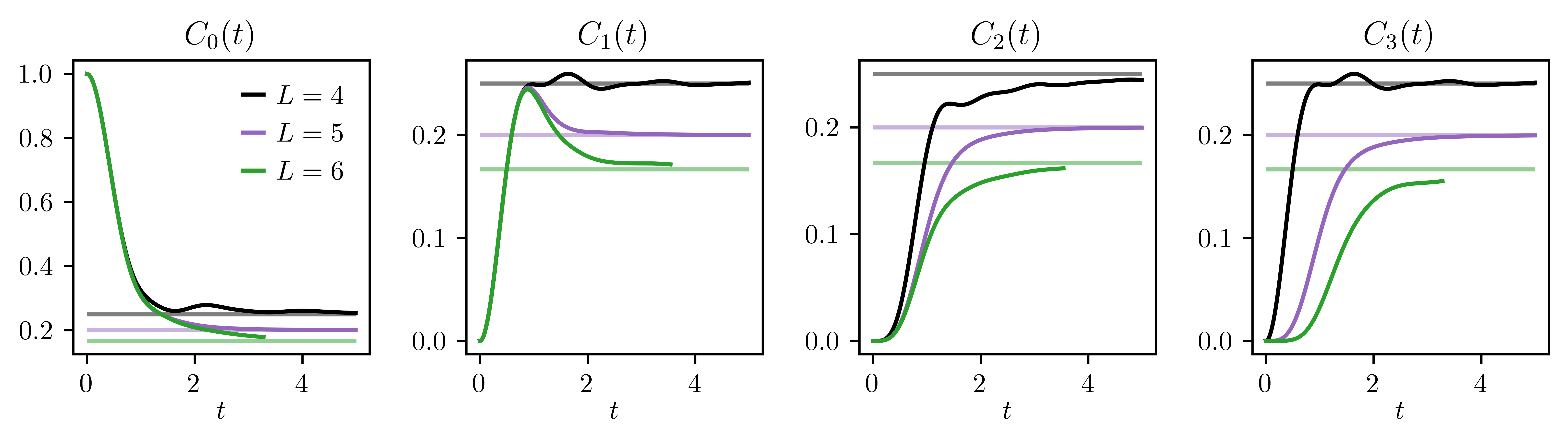}
    \caption{Two point functions for the Haar Ising chain obtained from the group word dynamics on large finite Cayley graphs of $G_L$; results are cut off self-consistently when edge effects occur. Horizontal lines are $L^{-1}$.}
    \label{fig:larger_L}
\end{figure*}

\begin{figure}
    \centering
    \includegraphics[width=0.6\columnwidth]{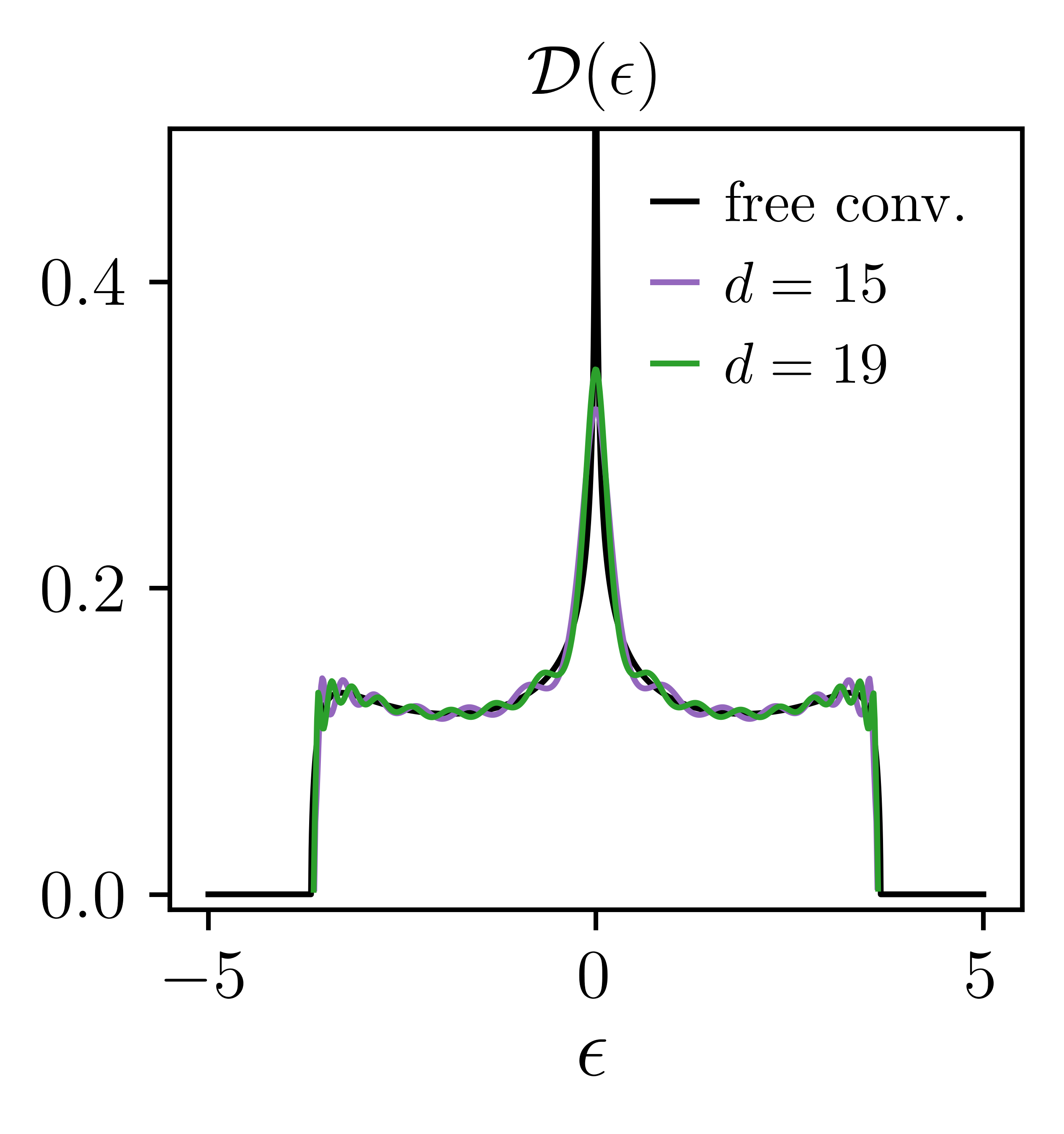}
    \caption{An approximate $L=4$ density of states obtained via the recursion method on finite clusters of the graphs of radius $d=15,19$. The $L=4$ graph in particular struggles to converge due to a peak (likely a van Hove singularity) at $\epsilon=0$. This method is compared against a different one (black curve), namely the subordination iteration method for free convolution, obtained in our previous paper \cite{pollock_2025}.}
    \label{fig:covergence}
\end{figure}

Here we show more data for two-point functions with $L=4,5,6$. These are Fig.~\ref{fig:larger_L}. We also show that the density of states for $G_4$ likely has a Van Hove singularity. This comes with oscillations of the $b_n$, and consequently oscillations of the approximated density of states around the true curve; see Fig.~\ref{fig:covergence}. Finally, we include some extensive data obtained from our computational group theory calculations, including the word growth, number of returning walks (a.k.a. co-growth) and some plots demonstrating the convergence (or lack thereof) of the Lanczos coefficients $b_n$ and estimates of their asymptotic values. The data are collected in Figs.~\ref{fig:Lanczos},\ref{tab:growth},\ref{tab:cogrowth}.

\begin{figure*}
    \setlength{\belowcaptionskip}{50pt}
    \centering
    \includegraphics[width=0.9\linewidth]{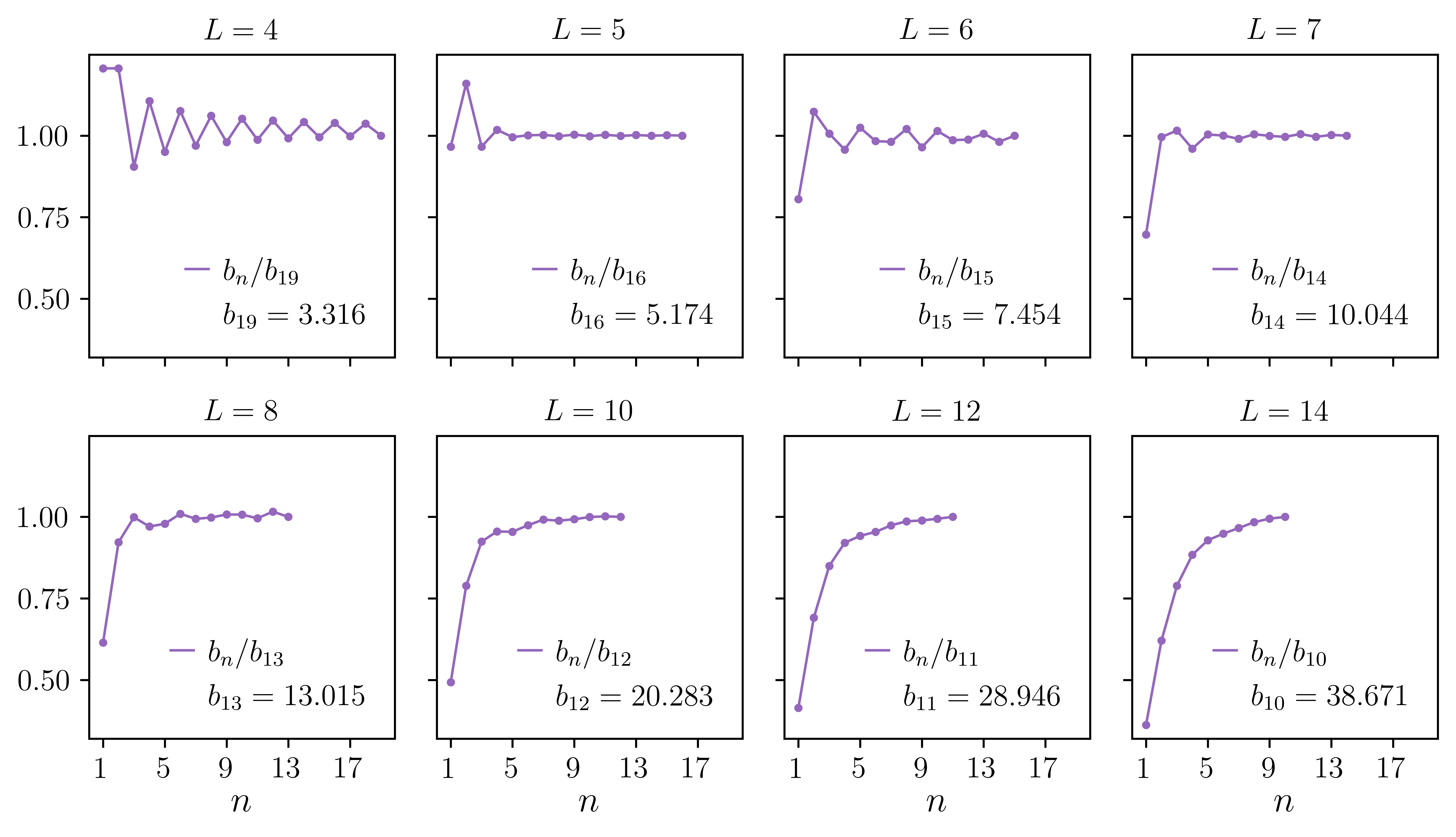}
    \caption{The Lanczos coefficients $b_n$ obtained via Lanczos tridiagonalization of the hopping matrix $\Delta$ on the largest accessible OBC cluster of the Cayley graph of $G_L$ up to graph radius $d=19,16,15,14,13,12,11,10$ for $L=4,5,6,7,8,10,12,14$, respectively. The quoted $b_d$ correspond to estimates of $b_\infty$, the converged result, which is clearly quite good for e.g. $L=5,7$ but less converged for other cases. This is the reason we focused on $L=5$ in the main text. }
    \label{fig:Lanczos}
    \begin{tabular}{r|*{13}{r}}
        $L$ & \multicolumn{13}{c}{$V(d)$ for $1\leq d \leq 13$} \\
        \midrule
        4 & 5 & 15 & 39 & 97 & 237 & 575 & 1391 & 3361 & 8117 & 19599 & 47319 & 114241 & 275805 \\
        5 & 6 & 21 & 61 & 166 & 441 & 1161 & 3046 & 7981 & 20901 & 54726 & 143281 & 375121 & 982086 \\
        6 & 7 & 28 & 90 & 264 & 744 & 2060 & 5660 & 15500 & 42388 & 115852 & 316564 & 864924 & 2363076 \\
        7 & 8 & 36 & 127 & 400 & 1191 & 3445 & 9815 & 27742 & 78086 & 219311 & 615252 & 1724997 & 4834929 \\
        8 & 9 & 45 & 173 & 583 & 1831 & 5527 & 16303 & 47433 & 136865 & 392925 & 1124557 & 3212367 & 9165535 \\
    \end{tabular}
    \caption{The exact growth function $V(d)$, or ``volume", i.e. number of vertices on the Cayley graph within graph distance $d$ of the origin $e$, for the family of groups $G_L$.}\label{tab:growth}
    \begin{tabular}{r|*{10}{r}}
        $L$ & \multicolumn{10}{c}{$\braket{e|\Delta^n|e}$ for $0\leq n \leq 20$ (even)} \\
        \midrule
        4 & 1 & 4 & 32 & 304 & 3136 & 33984 & 380672 & 4367360 & 51024896 & 604736512  \\
        5 & 1 & 5 & 55 & 755 & 11495 & 185915 & 3129055 & 54177065 & 958176215 & 17229311765  \\
        6 & 1 & 6 & 84 & 1536 & 31812 & 707976 & 16518060 & 398550984 & 9861416916 & 248830694376  \\
        7 & 1 & 7 & 119 & 2737 & 72835 & 2108757 & 64506071 & 2051455637 & 67164559483 & 2249185991437  \\
        8 & 1 & 8 & 160 & 4448 & 145904 & 5272128 & 202891408 & 8162154624 & 339384897968 & 14478049050752  \\
    \end{tabular}
    \caption{The exact co-growth function $\braket{e|\Delta^n|e}$, i.e. the number of walks of length $n$ that begin and end at $e$ on the Cayley graph, for the family of groups $G_L$.}
    \label{tab:cogrowth}
\end{figure*}

\section{Mazur bounds}\label{sec:mazur}

Mazur showed that long-time averages of autocorrelation functions can be bounded from below if one knows some integrals of motion of the dynamics \cite{mazur_1969,dhar_2021}. These bounds are expected to be tight if one knows and includes all integrals of motion. Let us suppose that for any particular realization, the only integrals of motion are powers of the Hamiltonian $H^k$. We write for the long time average
\begin{equation}
    \overline{F(t)} = \lim_{T\rightarrow \infty} \int_0^{T} \frac{dt}{T} F(t).
\end{equation}
Let us see what type of Mazur bound can be obtained with the set of integrals of motion $H^k$. For a single instance of the RMT we obtain the Mazur bound $\overline{C_0(t)}\geq B$ with
\begin{equation}\label{eq:fixed_realization}
    B = \sum_{nm} \frac{\text{tr}(h_iH^n)}{q^L} [C^{-1}]_{nm} \frac{\text{tr}(H^mh_i)}{q^L} ,
\end{equation}
where $C_{nm} = q^{-L}\text{tr}(H^{n+m})$. If we also assume that as $q\rightarrow \infty$, expectation values become self-averaging, and so all quantities can be replaced by their average values without correlations between the matrix elements of $C_{nm}$ and the quantities $q^{-L}\text{tr}(h_iH^n)$, we obtain
\begin{equation}
     B = \sum_{nm} \braket{V_i \Delta^n}_e [\widetilde{C}^{-1}]_{nm} \braket{\Delta^m V_i}_e ,
\end{equation}
with $\widetilde{C}_{nm} = \braket{\Delta^{n+m}}_e$, where we wrote $\braket{e|\cdot|e} = \braket{\cdot}_e$ for brevity. Then, we note that
\begin{equation}
    \braket{V \Delta^n}_e = L^{-1} \braket{\Delta^{n+1}}_e
\end{equation}
by symmetry, since $\braket{V_i \Delta^i}_e$ is the number of walks of length $i+1$ that begin in the particular direction $V_i$, and all of these directions are equivalent. Using also that
\begin{equation}
    \sum_{m=0}^\infty [\widetilde{C}^{-1}]_{nm}  \braket{\Delta^{m+1}}_e = \delta_{n,2}
\end{equation}
due to the particular Hankel form of $\widetilde{C}$, we find a Mazur bound of $B = L^{-1}$. This is also what would be obtained with just $H$.

Now, assuming that with high probability as $q\rightarrow \infty$ powers of $H$ are all the charges, we expect to have saturated the bound, i.e.
\begin{equation}
    \overline{C_0(t)} = L^{-1}
\end{equation}
with high probability as $q\rightarrow \infty$. Looking forward, a technical step is to show that correlations between various quantities in Eq.~\ref{eq:fixed_realization} can be neglected. 

\section{Braids lead to higher charges.}\label{sec:charges_in_involution}

Here, we demonstrate explicitly that for $L\geq 5$, the quantity $Q_4$ in the main text is also a conserved charge in involution with $Q_3$. We will employ some shorthand notation. For a fixed index $i$, we will write $g_{i+r}$ as $r$ and all expressions will have an implicit sum on $i$.
\begin{claim}\label{claim:overlapping}
    For a chain of length $L\geq 7$, including all overlapping densities, the commutator becomes in the shorthand notation
    \begin{multline}\label{eq:comm}
        [Q_4,Q_3] = \bigg[\big[[3,4],4 + 2\times 5\big],[1,2]+[2,3]+[3,4] \\
        +[4,5] + [5,6] + [6,7] \bigg].
    \end{multline}
\end{claim}
\begin{proof}
 Above, we started indexing the $Q_4$ density at $i+3$ simply to avoid writing negative symbols (we can begin indexing wherever we want by translation symmetry). Again, there is an implicit sum on $i$. We include all densities with overlapping support coming from $Q_4$ and $Q_3$, for which there are $6$ terms.
\end{proof}

To show that Eq.~\ref{eq:comm} vanishes, we will repeatedly make use of translation symmetry (valid since we are summing over all translations) and the group relations. A simplification which helps organize the calculation is the following. For the purposes of showing that the above commutator vanishes, there turns out to be an effective reflection symmetry around bond $i+4$. Since we are summing over translations and not reflections, this is not a real symmetry valid for any expression, but it turns out to be correct for the above commutator. So, we first show that
\begin{claim}\label{claim:reflection}
For the purposes of computing $[Q_4,Q_3]$ the following effective reflection symmetry around $4$ holds:
    \begin{multline}\label{eq:reflect}
        \bigg[\big[[3,4],4 + 2\times 5\big],[4,5] + [5,6] + [6,7] \bigg] \\
        = \bigg[\big[[5,4],4 + 2\times 3\big],[4,3] + [3,2] + [2,1] \bigg]
    \end{multline}   
where we have made the replacements $3\leftrightarrow 5$, $2\leftrightarrow 6$, and $1\leftrightarrow 7$. Here, $2 \times r$ stands for twice $g_{i+r}$.
\end{claim}
\begin{proof}
This follows after breaking the problem into a few steps. The first equality regards the leftmost term in the above expression. This has the symmetry on its own:
\begin{equation}
    \bigg[\big[[3,4],4 + 2\times 5\big],[4,5] \bigg] = \bigg[\big[[5,4],4 + 2 \times 3\big],[4,3] \bigg]. 
\end{equation}
The next step is to show that the rest of the terms also have the effective reflection together, i.e. that
\begin{multline}
    \bigg[\big[[3,4],4 + 2\times 5\big],[5,6] + [6,7] \bigg] \\
    = \bigg[\big[[5,4],4 + 2\times 3\big],[3,2] + [2,1] \bigg].
\end{multline}
Here, we note that the most offensive term can be transformed to make the symmetry manifest:
\begin{multline}
    \bigg[\big[[3,4],5\big],[6,7] \bigg] = \bigg[\big[[1,2],3\big],[4,5] \bigg] \\
    = \bigg[\big[[5,4],3\big],[2,1] \bigg].
\end{multline}
In the first equality, we used translation symmetry to shift, and in the second equality we used the Jacobi identity multiple times. At this point it remains to show only that
\begin{equation}
    \bigg[\big[[2,3],4 - 323 \big],[4,5] \bigg]  = \bigg[\big[[5,4],3 - 454 \big],[3,2] \bigg]
\end{equation}
which can be directly verified by expanding out all $20$ terms on either side and using the braid relations in the group to see the equality. This concludes the proof of Claim~\ref{claim:reflection}.
\end{proof}
Finally, we show that
\begin{claim}
    $[Q_4,Q_3] = 0$ for chains of length $L\geq 5$.
\end{claim}
\begin{proof}
Assuming $L\geq 7$ and using Claims~\ref{claim:overlapping},\ref{claim:reflection}, we can express
\begin{multline}\label{eq:shorthand}
    [Q_4,Q_3] = \bigg[\big[[3,4],4 + 2\times 5\big]-\big[[5,4],4 + 2\times 3\big],\\
    [1,2]+[2,3]+[3,4] \bigg].    
\end{multline}
By using the group relations, the difference of the two commutators in the first line already simplifies:
\begin{multline}
    \big[[3,4],4 + 2\times 5\big]-\big[[5,4],4 + 2\times 3\big] \\ 
    = 2 \times (3 - 5 - 343 + 454).
\end{multline}
We can also see by translation symmetry and braids that
\begin{equation}
    \big[3-5, [1,2]+[2,3]+[3,4]\big] = \big[3+4,[3,4]\big] = 0
\end{equation}
and finally that
\begin{multline}
    \big[-343 + 454, [1,2]+[2,3]+[3,4]\big] \\ = \big[343,[3,4]\big] = 0
\end{multline}
proving that braids are sufficient for another translation invariant local charge $Q_4$ in involution with $Q_3$.

The above result holds explicitly as written for $L\geq7$. However, Eq.~\ref{eq:reflect} also holds for $L=6$ if we read the symbols modulo $6$, in particular, $7\equiv 1$. None of the algebraic relations used in the rest of the proof change and so from this follows that $[Q_4,Q_3] = 0$ for the $L=6$ chain. Finally, one can directly verify without the idea of an effective reflection symmetry that the commutator still vanishes in the $L=5$ chain.
\end{proof}

\end{appendix}

\bibliography{references}

\end{document}